\newif\ifC
\Cfalse  
\ifC
    \documentclass[a4paper,UKenglish,cleveref, autoref, thm-restate]{lipics-v2021}
\else
    \documentclass{article}
\fi

\ifC
\title{Improved Approximation Algorithms for n-Pairs Shortest Paths}
\else
\title{Improved Approximation Algorithms for $n$-Pairs Shortest Paths}
\fi

\ifC
\else
\fi

\ifC
\author{Avi Kadria}
{Department of Computer Science, Bar-Ilan University, Ramat Gan 5290002, Israel}
{avi.kadria3@gmail.com}
{https://orcid.org/0000-0001-8449-3284}
{}

\author{Liam Roditty}
{Department of Computer Science, Bar-Ilan University, Ramat Gan 5290002, Israel}
{liam.roditty@biu.ac.il}
{https://orcid.org/0000-0002-5289-198X}
{Supported in part by BSF grants 2016365 and 2020356.}

\author{Virginia Vassilevska Williams}
{Department of Electrical Engineering and Computer Science and CSAIL, MIT, Cambridge, MA, USA}
{virgi@mit.edu}
{https://orcid.org/0000-0003-4844-2863}
{Supported in part by NSF CAREER Award 1651838, NSF Grants CCF-1909429 and CCF-2129139, BSF grants 2016365 and 2020356, a Google Research Fellowship, and a Sloan Research Fellowship.}

\authorrunning{A. Kadria, L. Roditty, and V. Vassilevska Williams}
\Copyright{Avi Kadria, Liam Roditty, and Virginia {Vassilevska Williams}}



\relatedversion{} 

\EventEditors{Philip Bille, Seth Pettie, and Sabine Storandt}
\EventNoEds{3}
\EventLongTitle{34th Annual European Symposium on Algorithms (ESA 2026)}
\EventShortTitle{ESA 2026}
\EventAcronym{ESA}
\EventYear{2026}
\EventDate{August 31--September 4, 2026}
\EventLocation{L'Aquila, Italy}
\EventLogo{}
\SeriesVolume{388}
\ArticleNo{64}

\else
\author{Avi Kadria\thanks{Department of Computer Science, Bar Ilan University, Ramat Gan 5290002, Israel. E-mail {\tt avi.kadria3@gmail.com}.} \and Liam Roditty\thanks{Department of Computer Science, Bar Ilan University, Ramat Gan 5290002, Israel. E-mail {\tt liam.roditty@biu.ac.il}. Supported in part by BSF grants 2016365 and 2020356.} \and Virginia Vassilevska Williams\thanks{Department of Electrical Engineering and Computer Science and CSAIL, MIT, Cambridge, MA, USA. E-mail {\tt virgi@mit.edu}. Supported in part by NSF CAREER Award 1651838, NSF Grants  CCF-1909429 and CCF- 2129139, BSF grants 2016365 and 2020356, a Google Research Fellowship and a Sloan Research Fellowship.}}

\fi

\ifC
\else
\usepackage[margin=1in]{geometry}
\fi
\usepackage{setspace}
\usepackage{multirow}
\usepackage{makecell}
\usepackage[bottom]{footmisc}

\usepackage{epigraph}
\usepackage{graphicx}
\usepackage{tikz}
\usetikzlibrary{arrows,automata,patterns}
\ifC\else\usepackage{enumitem}\fi
\usepackage{hyperref}
\usepackage{lipsum}

\usepackage[linesnumbered,ruled,vlined,boxed,algo2e]{algorithm2e}
\usepackage{nicefrac}
\usepackage{amsmath}
\usepackage{float}
\usepackage{mathtools}
\usepackage{amsthm}
\ifC\else
\usepackage{fullpage}
\fi
\usepackage{amsfonts}
\usepackage{clipboard}
\ifC
\else
\usepackage[nameinlink]{cleveref}
\newtheorem{theorem}{Theorem}

\newtheorem{lemma}{Lemma}

\newtheorem{corollary}[theorem]{Corollary}

\fi
\newtheorem{cclaim}{Claim}[lemma]

\newtheorem{subclaim}{Claim}[cclaim]

\newtheorem{question}{Question}

\Crefname{subclaim}{Claim}{Claims}
\Crefname{enumi}{(item)}{(items)}

\newcommand{\Reminder}[1]{

\vspace{0.5em}
\noindent\textbf{Reminder of~\autoref{#1}.} \textit{\Paste{#1}}
\vspace{0.5em}

}

\usepackage{xcolor}
\definecolor{DarkGreen}{RGB}{1,50,32}
\usepackage{pgfplots}
\pgfplotsset{compat=1.18}

\usepackage{silence}
\usepackage{chngcntr}
\counterwithin*{equation}{section}
\counterwithin*{equation}{subsection}

\begin{document}
\ActivateWarningFilters[pdftoc]
\newcommand{\defeq}{:=}
\newcommand{\eps}{\varepsilon}

\newcommand{\liam}[1]{{\color{red} \textbf{Liam}: #1}} 
\newcommand{\avi}[1]{{\color{purple} \textbf{Avi}: #1}} 
\newcommand{\virgi}[1]{{\color{red} \textbf{Virgi}: #1}} 
\newcommand{\new}[1]{{\color{blue} #1}}

\newcommand{\blue}[1]{{\color{blue}#1}}

\newcommand{\ReturnCode}{\textbf{return}}
\newcommand{\codestyle}[1]{\texttt{#1}}
\newcommand{\Initialize}{\mbox{\codestyle{Initialize}}}
\newcommand{\Cycle}{\mbox{\codestyle{Cycle}}}
\newcommand{\ADO}{\mbox{\codestyle{ADO}}}
\newcommand{\hADO}{\mbox{\codestyle{hADO}}}
\newcommand{\Query}{\mbox{\codestyle{.Query}}}
\newcommand{\THZQuery}{\mbox{\codestyle{TZQuery}}}
\newcommand{\ADOQuery}{\mbox{\codestyle{ADO.Query}}}
\newcommand{\ConstructADO}{\codestyle{ConstructADO}}
\newcommand{\FastAPSP}{\codestyle{FastAPSP}}
\newcommand{\Construct}{\mbox{\codestyle{Construct}}}
\newcommand{\Spanner}{\codestyle{Spanner}}
\newcommand{\Emulator}{\codestyle{Emulator}}
\newcommand{\RT}{\mbox{\codestyle{RT}}}
\renewcommand{\L}{\mbox{\codestyle{L}}}
\newcommand{\CycleOdd}{\codestyle{CycleOdd}}
\newcommand{\BallOrCycle}{\codestyle{BallOrCycle}}
\newcommand{\ClusterOrCycleBounded}{\codestyle{ClusterOrCycleBounded}}
\newcommand{\ApproximateANSC}{\codestyle{ApproximateANSC}}
\newcommand{\ApproximateANSCUnweighted}{\codestyle{ApproximateANSCUnweighted}}

\newcommand{\ClusterOrCycle}{\codestyle{ClusterOrCycle}}
\newcommand{\SimpleCycle}{\codestyle{SimpleCycle}}
\newcommand{\Next}{\codestyle{Next}}
\newcommand{\Sample}{\codestyle{Sample}}
\newcommand{\Dijkstra}{\codestyle{Dijkstra}}
\newcommand{\Preprocess}{\codestyle{Preprocess}}
\newcommand{\HashTable}{\codestyle{HashTable}}
\newcommand{\Heap}{\codestyle{Heap}}
\newcommand{\RelaxNext}{\codestyle{RelaxNext}}
\newcommand{\ODD}{_{\mathrm{ODD}}}

\newcommand{\PreprocessGraph}{\codestyle{Initialize}}
\newcommand{\Route}{\codestyle{Route}}
\newcommand{\TreeRoute}{\codestyle{TreeRoute}}
\newcommand{\N}{\mathbb{N}}
\newcommand{\MinCycle}{\codestyle{MinCycle}}
\newcommand{\Ball}{\codestyle{Ball}}
\newcommand{\DistanceOracle}{\codestyle{TZ-DistanceOracle}}
\newcommand{\SparseOrCycle}{\codestyle{SparseOrCycle}}
\newcommand{\Intersection}{\codestyle{Intersection}}
\newcommand{\CycleAdditive}{\codestyle{CycleAdditive}}
\newcommand{\GenerateSi}{\codestyle{ComputeS}}

\newcommand{\codeNull}{\codestyle{null}}
\newcommand{\codeYes}{\codestyle{Yes}}
\newcommand{\codeNo}{\codestyle{No}}
\newcommand{\codeAnd}{~ \mathrm{and} ~}
\newcommand{\codeOr}{~ \mathrm{or} ~}
\newcommand{\wt}{\ell}
\newcommand{\Cl}{CL}
\newcommand{\CL}{CL}
\newcommand{\cl}{c\ell}
\newcommand{\etal}{\textit{et~al.}}
\newcommand{\EQ}{\;=\;}
\newcommand{\GE}{\;\ge\;}
\newcommand{\Ot}{\tilde{O}}
\newcommand{\stactri}{\stackrel\triangle}

\newcommand{\tPSP}{$t$-PSP\xspace}
\newcommand{\nPSP}{$n$-PSP\xspace}
\newcommand{\ParameterizedNPSP}{\codestyle{\codestyle{Parameterized-nPSP}}}
\newcommand{\EE}{\mathbb{E}}
\newcommand{\RR}{\mathbb{R}}
\newcommand{\nPSPTheorem}[2]{
    Let #1, and let $k\ge 4$ be an integer parameter. Let $I$ be a set of $n$ vertex pairs. There is an algorithm that computes in $\Ot(mn^{1/k}+n^{1+2/k})$ time, for every $\pair{u,v}\in I$, an estimate $\hat{d}(u,v)$ such that $d(u,v) \le \hat{d}(u,v) \le {#2}$.
}

\newcommand{\nPSPTheoremt}[2]{
    Let #1, and let $k\ge 4$ be an integer parameter. Let $I$ be a set of $n$ vertex pairs. There is an algorithm that computes in $\Ot(m^{1-1/k}n^{2/k})$ time, for every $\pair{u,v}\in I$, an estimate $\hat{d}(u,v)$ such that $d(u,v) \le \hat{d}(u,v) \le {#2}$.
}

\ifC
\renewcommand{\paragraph}[1]{\textbf{#1}}
\else
\fi

\DeclarePairedDelimiter{\ceil}{\lceil}{\rceil}
\DeclarePairedDelimiter{\floor}{\lfloor}{\rfloor}
\DeclarePairedDelimiter{\pair}{\langle}{\rangle}

\maketitle

\thispagestyle{empty}
\begin{abstract}
Let $G = (V, E)$ be a graph with $n = |V|$ nodes and $m = |E|$ edges. The $t$-Pairs Shortest Paths problem, introduced by Cohen
[FOCS'93; SICOMP'99], asks to approximate the distances between $t$ prespecified pairs of vertices. Recently, this problem has received renewed attention, particularly in the case where $t = \Theta(n)$: the $n$-Pairs Shortest Paths problem. In this setting, new algorithms and conditional lower bounds have been developed by Dalirrooyfard, Jin, Vassilevska Williams, and Wein [FOCS'22], and Chechik, Hoch, and Lifshitz [SODA'25].

In this paper, we present the first algorithm for the $n$-Pairs Shortest Paths problem in \textit{weighted} undirected graphs that achieves a $(2 - \alpha)k$-approximation, for constant $\alpha > 0$, that runs in $\tilde{O}(mn^{1/k} + n^{1 + 2/k})$ time. Specifically, we present a $1.622k$-approximation, improving upon the $(2k - 3)$-approximation of Chechik, Hoch, and Lifshitz [SODA'25] for graphs that are not super sparse, which answers in the affirmative the open question posed by them.
We also develop improved approximation algorithms with better tradeoffs for unweighted graphs and dense weighted graphs that improve upon the results of Dalirrooyfard \etal~and Chechik, Hoch, and Lifshitz.

Our main technical contribution is the new \textit{heavy-edge} technique. Using this technique, we transform an algorithm with an approximation guarantee that depends on $W_{uv}$, the weight of the heaviest edge on the shortest path between $u$ and $v$, into an algorithm with purely multiplicative approximation that does not depend on $W_{uv}$. 

\end{abstract}
\clearpage
\newpage
\thispagestyle{empty}

\section{Introduction}

Let $G = (V, E)$ be a graph with $n = |V|$ nodes and $m = |E|$  edges, and let $d(u,v)$ be the distance between $u$ and $v$, for every $u,v\in V$. 
Let $I \subseteq V \times V$ be a set of $t$ prespecified pairs of vertices. In the  $t$-Pairs Shortest Paths (\tPSP) problem, the goal is to efficiently compute (or approximate) the shortest path distance $d(u, v)$ for each pair $(u, v) \in I$.
The \tPSP~problem was first studied in the 1990s
by Cohen~\cite{DBLP:journals/siamcomp/Cohen98} and then
by Aingworth, Chekuri, Indyk, and Motwani~\cite{ACIM99}. Aingworth \etal~\cite{ACIM99} presented an algorithm that computes a $(1,2)$-approximation\footnote{An estimation $\hat{d}(u,v)$ of $d(u,v)$ is an
$(\alpha,\beta)$-approximation if $d(u,v)\leq \hat{d}(u,v) \leq \alpha d(u,v) + \beta$. We denote purely multiplicative approximation ($\beta = 0$) with $(\alpha)$-approximation.} for unweighted, undirected graphs in $\tilde{O}(n^{1.5} \sqrt{t} + n^2)$ time. \footnote{$\Ot$ omits polylogarithmic factors.}

A natural approach for approximating the \tPSP~problem is to use the celebrated Approximate Distance Oracles (ADOs) of Thorup and Zwick~\cite{DBLP:journals/jacm/ThorupZ05}. These oracles are constructed for weighted undirected graphs in $\tilde{O}(mn^{1/k})$ time, use $\Ot(n^{1+1/k})$ space, and guarantee a $(2k - 1)$-approximation for distances between any pair of nodes. 
Once built, they can answer each distance query in $O(k)$ time, so solving the \tPSP~problem using an ADO is straightforward and requires $\tilde{O}(mn^{1/k} + t)$ total time.

ADOs are primarily designed to optimize \textit{space efficiency} while supporting distance queries for \textit{all} pairs of vertices. In contrast, the \tPSP~problem only requires distance queries between a specific set of $t$ vertex pairs, and the main objective is to minimize \textit{running time}, not space. 
This distinction suggests that the generality of ADOs may be unnecessary in the \tPSP~setting, and that one can hope to design faster algorithms that achieve improved time-approximation tradeoffs for \tPSP.

Interestingly, although the \tPSP~problem was first introduced in the 1990s, it remained largely unexplored for over two decades. In the meantime, closely related problems, primarily focused on space-efficient representations such as pairwise spanners, distance preservers, and reachability preservers, received significantly more attention (see, for example,~\cite{CoppersmithElkin2006,AbboudBodwin2018,DBLP:conf/focs/KoganP22}).

Recently, the \tPSP problem received renewed attention in the field of fine-grained complexity.
Abboud~\etal~\cite{DBLP:conf/stoc/AbboudBKZ22, DBLP:conf/stoc/AbboudBF23} and independently Jin and Xu~\cite{jin2023removing} established several conditional lower bounds for \tPSP based on the $3$SUM hypothesis. In particular, Jin and Xu~\cite{jin2023removing} proved that there is no $\Ot(m^{1 + \frac{1}{2k + 1} - \varepsilon})$-time algorithm that achieves a $(k - \delta)$-approximation for the $m$-PSP problem in graphs where $m = \Theta(n^{1 + \frac{1}{2k - 2}})$.

Dalirrooyfard~\etal~\cite{DBLP:conf/focs/DalirrooyfardJW22} studied the \nPSP problem, which is a special case of \tPSP with $t=n$. They established several conditional lower bounds.
Among their results, they showed that, assuming the combinatorial $4k$-Clique Hypothesis,\footnote{The combinatorial $4k$-Clique Hypothesis asserts that there is no $O(n^{k - \varepsilon})$-time combinatorial algorithm for detecting a $k$-clique, for any $\varepsilon > 0$.} there is no $(1 + 1/k - \delta)$-approximation algorithm for \nPSP that runs in time $\tilde{O}(m^{2 - \frac{2}{k+1}} n^{\frac{1}{k+1} - \varepsilon})$. This lower bound \textit{matches} an upper bound that follows from the work of Agarwal~\cite{DBLP:conf/esa/Agarwal14} on distance oracles with non-constant query time.

\cite{DBLP:conf/focs/DalirrooyfardJW22} also considered \nPSP algorithms for unweighted undirected graphs. They presented a $(2 + \varepsilon, \beta)$-approximation algorithm for unweighted undirected graphs running in $\tilde{O}(m + n^{3/2 + \varepsilon})$ time, which nearly matches their lower bound for sparse graphs.
For general approximations, they \cite{DBLP:conf/focs/DalirrooyfardJW22} designed a $(2k - 2,1)$-approximation algorithm running in $\tilde{O}(mn^{1/k})$ time, improving the classic $(2k - 1)$-approximation of Thorup and Zwick for the special case of \nPSP. 

Very recently, Chechik, Hoch, and Lifshitz~\cite{DBLP:conf/soda/ChechikHL25} established new upper bounds for \nPSP~by presenting several approximation algorithms with improved tradeoffs. In particular, they present an algorithm that computes a $(2k - 3)$-approximation for \nPSP~in \textit{weighted} undirected graphs in $\Ot(mn^{1/k})$ expected time, for any integer $k \geq 3$. 
This improves upon the $(2k - 2,1)$-approximation of Dalirrooyfard \etal~\cite{DBLP:conf/focs/DalirrooyfardJW22} that worked only in unweighted graphs.
In light of their $(2k - 3)$-approximation for \nPSP~in weighted graphs, Chechik, Hoch, and Lifshitz posed the following question: 

\begin{question}[\cite{DBLP:conf/soda/ChechikHL25}]\label{Q-less-2k}
Is it possible to break the multiplicative  $(2k - O(1))$-approximation barrier for \nPSP with an algorithm running in $\tilde{O}(mn^{1/k})$ time?
\end{question}

The main result of this paper is an affirmative answer to this question for \textit{weighted} graphs that are not super sparse. Specifically, we present a $1.622k$-approximation algorithm for \nPSP~in weighted graphs that runs in $\tilde{O}(mn^{1/k} + n^{1 + 2/k})$ time. 

\begin{theorem}\label{T-nPSP-Weighted-Sparse}\Copy{T-nPSP-Weighted-Sparse}{
    \nPSPTheorem{$G=(V, E, \wt)$ be a weighted undirected graph, where $\wt:E\rightarrow \RR_{\ge0}$}{1.622k\cdot d(u,v)}}
\end{theorem}

This theorem improves on the $(2k - 3)$-approximation algorithm of~\cite{DBLP:conf/soda/ChechikHL25} for all $k \geq 8$.
We also note that our new techniques can be used to improve the dense-graph \nPSP algorithm of~\cite{DBLP:conf/soda/ChechikHL25}, which runs in $\tilde{O}\!\left(m^{1-1/k}n^{2/k}\right)$ time and achieves a $(2k+1)$-approximation. 

\begin{corollary}
    \nPSPTheoremt{$G=(V, E, \wt)$ be a weighted undirected graph, where $\wt:E\rightarrow \RR_{\ge0}$}{(1.622k+4)\cdot d(u,v)}
\end{corollary}

\paragraph{Unweighted \nPSP.}
In the \textit{unweighted} setting, Chechik, Hoch, and Lifshitz~\cite{DBLP:conf/soda/ChechikHL25} answered~\Cref{Q-less-2k} affirmatively for graphs that are not super sparse. They presented a $(\ceil{\tfrac{4k}{3}} - 1, \ceil{\tfrac{4k}{3}}- 1)$-approximation algorithm for \nPSP~with a running time of $\tilde{O}(mn^{1/k} + n^{1 + 2/k})$.
We improve upon this result by reducing the multiplicative approximation error and eliminating the additive error. 
\begin{theorem}\label{T-k-stretch-k=4-5}\Copy{T-k-stretch-k=4-5}{
    \nPSPTheorem{$G=(V, E)$ be an unweighted undirected graph}{\ceil{\tfrac{4k}{3}-\tfrac{5}{3}} d(u,v)}}
\end{theorem}

We remark that our improved approximation scheme yields a $k$-approximation algorithm for \nPSP~with running time $\tilde{O}(mn^{1/k} + n^{1 + 2/k})$ for $4 \leq k \leq 5$. This naturally raises the question of whether a $k$-approximation algorithm with the running time $\tilde{O}(mn^{1/k} + n^{1 + 2/k})$ exists for \emph{all} values of $k$. 

As a byproduct of our improved algorithm for unweighted graphs, we extend the result of~\cite{DBLP:conf/soda/ChechikHL25} to weighted graphs. Specifically, we present a $(\lceil \tfrac{4k}{3} \rceil - 1,\, (\lceil \tfrac{4k}{3} \rceil - 1) \cdot W_{uv})$-approximation algorithm for \nPSP, where $W_{uv}$ denotes the weight of the heaviest edge on a shortest $u$-$v$ path. The algorithm runs in $\tilde{O}(mn^{1/k} + n^{1 + 2/k})$ time. 

\begin{theorem}\label{T-nPSP-weighted-Sparse-Wuv} \Copy{T-nPSP-weighted-Sparse-Wuv}{
    \nPSPTheorem{$G=(V, E, \wt)$ be a weighted undirected graph, where $\wt:E\rightarrow \RR_{\ge0}$}{(\ceil{\tfrac{4k}{3}}-1) d(u,v)+(\ceil{\tfrac{4k}{3}}-1)W_{uv}}}
\end{theorem}

\paragraph{Linear time algorithms for \tPSP.}
In the setting of \textit{dense} unweighted graphs, Dalirrooyfard~\etal~\cite{DBLP:conf/focs/DalirrooyfardJW22} presented a $((2k - 1)(2k - 2))$-approximation algorithm for \nPSP, running in $\tilde{O}(m + n^{1 + 2/k})$ time, which is linear for graphs with $m=\Omega(n^{1+2/k})$. By combining their algorithm with the $(2k - 3)$-approximation algorithm of Chechik~\etal~\cite{DBLP:conf/soda/ChechikHL25}, one can achieve a $((2k - 1)(2k - 3))$-approximation algorithm for weighted graphs in the same running time. 

We improve upon this result by presenting a $(k^2 + 2k)$-approximation algorithm for the more general \tPSP~problem in weighted graphs, with running time $\tilde{O}(m + n^{1 + 2/k} + t)$, which is linear for graphs with $m=\Omega(n^{1+2/k})$.
\begin{theorem}\label{T-nPSP-dense-k^2+2}\Copy{T-nPSP-dense-k^2+2}{
    Let $G=(V, E)$ be a weighted undirected graph, and let $k\ge 3$ be an integer parameter. Let $I$ be a set of $n$ vertex pairs. There is an algorithm that computes in $\Ot(m+n^{1+2/k})$ time, for every $\pair{u,v}\in I$, an estimate $\hat{d}(u,v)$ such that $d(u,v) \le \hat{d}(u,v) \le (k^2 + 2k)\cdot d(u,v)$.}
\end{theorem}

We remark that, in addition to achieving a better approximation in the same running time than \cite{DBLP:conf/focs/DalirrooyfardJW22,DBLP:conf/soda/ChechikHL25}, our algorithm can also answer more than $n$ distance queries in the same time.
All of our results are summarized in~\Cref{tab:npsp-weighted}.

\begin{table}[t]
    \centering
    \begin{tabular}{|c|c|c|c|}
        \hline
        Time  & Estimation &  Ref. & Comment \\
        \hline\hline
         $\Ot(mn^{1/k})$ & $(2k-1)\delta$ &  \cite{DBLP:journals/jacm/ThorupZ05} & \\
         \hline
         $\Ot(mn^{1/k})$ & $(2k-2)\delta$ &  \cite{DBLP:conf/focs/DalirrooyfardJW22} & \\
         \hline
         $\Ot(mn^{1/k})$ & $(2k-3)\delta$ &  \cite{DBLP:conf/soda/ChechikHL25} & \\
         \hline
         $\Ot(mn^{1/k}+n^{1+2/k})$ & $1.622k\delta$ &  \Cref{T-nPSP-Weighted-Sparse} & \\
                  \hline
                  $\Ot(mn^{1/k}+n^{1+2/k})$ & $\ceil{\tfrac{4k}{3}-1} \delta+\ceil{\tfrac{4k}{3}-1}W_{uv}$ &  \Cref{T-nPSP-weighted-Sparse-Wuv} & \\

        \hline\hline
         $\Ot(mn^{1/k}+n^{1+2/k})$ & $\ceil{\tfrac{4k}{3}-1}\delta + \ceil{\tfrac{4k}{3}-1}$ &   \cite{DBLP:conf/soda/ChechikHL25} & Unweighted graphs\\
         \hline
         $\Ot(mn^{1/k}+n^{1+2/k})$ & $\ceil{\tfrac{4k}{3}-\tfrac{5}{3}} \delta$ & \Cref{T-k-stretch-k=4-5} & Unweighted graphs\\
        \hline\hline

         $\Ot(m+n^{1+2/k})$ & $(4k^2-8k+3)\delta$ &  \cite{DBLP:conf/focs/DalirrooyfardJW22}+\cite{DBLP:conf/soda/ChechikHL25} & \\
         \hline
         $\Ot(m+n^{1+2/k} + t)$ & $(k^2+2k)\delta$ &  \Cref{T-nPSP-dense-k^2+2} & For \tPSP \\
    \hline
         
    \end{tabular}
    \caption{\nPSP results in \textit{weighted}, unless specified otherwise, undirected graphs with real edge weights. $\delta=d(u,v)$, and $W_{uv}=\max_{(w,z)\in P(u,v)}\wt(w,z)$.
    }
    \label{tab:npsp-weighted}
\end{table}

\subsection{Related shortest paths problems}
\paragraph{ADOs with non-constant query time.}
A natural approach to \nPSP is to construct an ADO and then query the ADO with the $n$ input pairs.
The total running time is the ADO construction time plus the cost of $n$ queries.

Following this approach with the Thorup-Zwick ADO~\cite{DBLP:journals/jacm/ThorupZ05}, yields a total running time of $\tilde{O}(mn^{1/k} + nk)$. Notice that the $\tilde{O}(nk)$ cost of the queries is negligible when compared to the $\tilde{O}(mn^{1/k})$ cost of the construction. This imbalance suggests that by using more time in the query, improved approximation guarantees might be achieved without increasing the overall runtime. 

To keep the total running time $\tilde{O}(mn^{1/k})$, the per-query time may be as large as $\tilde{O}(\frac{m}{n}n^{1/k})$. Recently, Kadria and Roditty~\cite{kadria2025new} studied ADOs with query time $\tilde{O}(\frac{m}{n}n^{1/k})$. They constructed a $(2k-5)$-ADO with query time $\tilde{O}(\frac{m}{n}n^{1/k})$ and space $\tilde{O}(m+n^{1+1/k})$. 
They asked whether one can obtain an approximation strictly better than $2k-O(1)$ using the same space and query time.
However, ADOs typically prioritize space efficiency at the cost of increased construction time. Indeed, the construction time in~\cite{kadria2025new} is $\Ot(mn^{3/k})$ time, which exceeds our target of $\Ot(mn^{1/k})$ time. Therefore, the two questions are distinct: an affirmative answer to the question of~\cite{kadria2025new} does not answer \Cref{Q-less-2k} due to the larger construction time, and our affirmative answer to \Cref{Q-less-2k} does not answer the question of~\cite{kadria2025new}  because of our larger space usage. 

\paragraph{Single-Source and Multi-Source Shortest paths.}
The Single-Source Shortest Paths (SSSP) and Multi-Source Shortest Paths (MSSP) problems are fundamental shortest paths problems. In SSSP, given a source vertex, the goal is to compute distances to all $n$ vertices in the graph. In MSSP, one is given a set of $n^{\alpha}$ sources, for $0<\alpha<1$, and the goal is to compute distances from each source to all $n$ vertices in the graph.

Unlike the above two problems, \nPSP allows instances in which the query set involves $\Omega(n)$ distinct sources, and these instances appear to be the hardest. Indeed, existing conditional lower bounds for \nPSP~\cite{DBLP:conf/focs/DalirrooyfardJW22, DBLP:conf/stoc/AbboudBKZ22, DBLP:conf/stoc/AbboudBF23, jin2023removing} 
 hold only when the set of queries contains  $\Omega(n)$ disjoint  vertex pairs.

Consequently, the hard instances of \nPSP and \tPSP have high vertex diversity. From this perspective, SSSP and MSSP represent structurally 
restricted cases that centralize source vertices, allowing for specialized, 
faster algorithms. This fundamental difference motivates the study of \nPSP as a distinct problem in its own right, necessitating novel algorithmic frameworks and techniques.

\paragraph{Paper organization.}
In the next section, we provide a technical overview of the paper. In \Cref{S-4k/3-unweighted} we present as a warm-up a simple $(\ceil{\tfrac{4k}{3}}-1,(\ceil{\tfrac{4k}{3}}-1)\ODD)$-approximation \nPSP algorithm, that we will use as a base algorithm in the later sections. \Cref{S-nPSP-weighted} is the main technical section. The section starts with adapting the results of \Cref{S-4k/3-unweighted} to weighted graphs with an approximation that depends on $W_{uv}$, where $W_{uv}$ denotes the weight of the heaviest edge on a shortest $u$-$v$ path. We proceed and develop the heavy-edge technique that allows us to achieve a $1.622k$-approximation \nPSP algorithm for general weighted undirected graphs, which is the main result of this paper.
In \Cref{S-dense}, we address the \tPSP problem in dense weighted graphs, and present our improved running time for \tPSP, which not only improves the running time for previous \nPSP algorithms, but supports more distance queries in the same time. 

Finally, in \Cref{S-k-stretch-small-k}, we revisit the \nPSP algorithm from \Cref{S-4k/3-unweighted}, and using tighter analysis and slight modifications to the algorithm, we improve the multiplicative approximation to $(\ceil{\tfrac{4k}{3}-\tfrac{5}{3}})$, and remove the additive error entirely.

\section{Preliminaries}
Let $G=(V,E)$ be an undirected graph with $n=|V|$ vertices and $m=|E|$ edges. Throughout the paper, we consider both unweighted graphs and weighted graphs with non-negative real edge weights. 
Let $u,v\in V$. The distance $d(u,v)$ between $u$ and $v$ is the length of a shortest path between $u$ and $v$. 
Let $P(u,v)$ be a shortest path between $u$ and $v$; when there are several, we take one that minimizes the maximum edge weight. In weighted graphs, let $W_{uv}$ be the weight of the longest edge in $P(u,v)$, i.e. $W_{uv}=\max_{(w,z)\in P(u,v)}\wt(w,z)$.\footnote{When there are multiple shortest paths $P$ between $u$ and $v$, $W_{uv}$ denotes $\min_P\max_{(w,z)\in P}\wt(w,z)$. When the path $P(u,v)$ is clear from context, $W_{uv}$ denotes the weight of the longest edge in $P(u,v)$.}
Let $x\ODD$ be $x\cdot (d(u,v)\pmod 2)$.
The distance $d(u,X)$ between $u$ and $X$ is the distance between $u$ and the closest vertex to $u$ from $X$, that is, 
$d(u, X)= \min_{x\in X}(d(u,x))$. Let $p(u, X)=\arg \min_{x\in X}(d(u,x))$ (ties are broken in favor of the vertex with a smaller identifier). 

Following many previous algorithms (\cite{DBLP:journals/jacm/ThorupZ05,DBLP:conf/soda/ChechikHL25,DBLP:conf/focs/DalirrooyfardJW22, DBLP:journals/siamcomp/BaswanaK10} and more.) we let the bunch $B(u,X,Y)$ be $\{y\in Y \mid d(u,y)<d(u,X)\} \cup \{p(u, X)\}$, where $X,Y\subseteq V$. The bunch can be viewed as a set of vertices closer to $u$ than its closest pivot.

In our algorithms, as well as in many distance-related algorithms (\nPSP, distance oracles, all-pairs shortest paths, multiple-source shortest paths, spanners, etc.),  there is a vertex hierarchy $V=A_0\supset A_1 \supset \ldots \supset A_k=\emptyset$, where $A_{i+1}$ contains every vertex from $A_i$ with probability $n^{-1/k}$.
Let $B_i(u)$ be $B(u,A_{i+1},A_i)$, $p_i(u)=p(u,A_i)$, and $h_i(u)=d(u,A_i)$. Let $B(u)=\cup_i B_i(u)$.

In their seminal work, Thorup and Zwick \cite{DBLP:journals/jacm/ThorupZ05} also presented an efficient algorithm for computing $B(u)$ for every $u\in V$, as described in the following lemma.
\begin{lemma}[\cite{DBLP:journals/jacm/ThorupZ05}]\label{L-THZ-construction}
    There is an $\Ot(mn^{1/k})$ time algorithm that computes $B(u)$ and the distances from $u$ to every $w\in B(u)$, for every $u\in V$.
\end{lemma}

The following classic lemma from \cite{DBLP:journals/jacm/ThorupZ05} bounds $h_i(u)$ in the case that $p_j(u) \notin B(v)$ and $p_j(v) \notin B(u)$ for every $0 < j < i$. This lemma plays a crucial role in the design of their distance oracle.
\begin{lemma}[\cite{DBLP:journals/jacm/ThorupZ05}] \label{L-Bound-p-i-d}
     Let $u,v\in V$ and let $0 < i \leq k-1$. If $p_j(u) \notin B(v)$ and $p_j(v) \notin B(u)$ for every $0 < j < i$, then
    \[
    h_i(u) \leq i \cdot d(u, v) \quad \text{and} \quad h_i(v) \leq i \cdot d(u, v).
    \]
\end{lemma}


Let $\THZQuery(u,v)$ be the query of the distance oracle of \cite{DBLP:journals/jacm/ThorupZ05}. $\THZQuery(u,v)$ finds the minimum $i$ such that $p_i(u)\in B(v)$ or $p_i(v)\in B(u)$, and then returns $\min(d(u,p_i(v))+d(p_i(v),v), d(u,p_i(u))+d(p_i(u),v))$ as the estimation.
Using \Cref{L-Bound-p-i-d},  \cite{DBLP:journals/jacm/ThorupZ05} proved:
\begin{lemma}[\cite{DBLP:journals/jacm/ThorupZ05}] \label{L-THZ05-Q-Correctness-private}
    For every $0\le i\le k-1$ it holds that $\THZQuery(u,v) \le (2k-2i-1)d(u,v)+2h_i(u)$.
\end{lemma}

The following lemma is implicit in \cite{DBLP:conf/spaa/ThorupZ01}, and explicit in \cite{DBLP:conf/soda/ChechikHL25}. Using this lemma, one can bound $h_i(u)$ in the case that $p_{i-1}(v)\notin B(u)$. We add its proof for completeness since it is used in the correctness of \Cref{T-nPSP-Weighted-Sparse}.
\begin{lemma}[\cite{DBLP:conf/spaa/ThorupZ01,DBLP:conf/soda/ChechikHL25}]
    \label{L-Bound-delta-2d}
    Let $u,v\in V$, if $p_{i-1}(v)\notin B(u)$ then $d(v, p_{i}(v)) \le d(v, p_{i-1}(v)) + 2d(u, v)$.
\end{lemma}
\begin{proof}
    From the definition of $p_i(v)$, we know that it is the closest vertex to $v$ in $A_i$, and since $p_i(u)\in A_i$ we get that  
    $d(v, p_i(v)) \leq d(v, p_i(u))$.
    From the triangle inequality, it follows that
    $d(v, p_i(u)) \leq d(v, u) + d(u, p_i(u))$.
    By the assumption of the lemma, we know that $p_{i-1}(v) \notin B(u)$. Therefore,    $d(u, p_i(u)) \leq d(u, p_{i - 1}(v))$.
    From the triangle inequality, it follows that
    $d(u, p_{i - 1}(v)) \leq d(u, v) + d(v, p_{i - 1}(v))$.
    Overall, we get that:
    \begin{align*}
        d(v, p_i(v)) &\leq d(v, p_i(u)) \leq d(v, u) + d(u , p_i(u)) \leq d(v ,u) + d(u, p_{i - 1}(v)) 
        \\&\leq d(v , u) + d(u , v) + d(v, p_{i - 1}(v)) = 2d(u , v) + d(v, p_{i - 1}(v)), \text{ as required.}
    \end{align*}
\end{proof}

Thorup and Zwick \cite{DBLP:journals/jacm/ThorupZ05} proved the following bound on the size of the bunches. 
\begin{lemma}[Lemma 3.5 in~\cite{DBLP:journals/jacm/ThorupZ05}]
    $|B(u)| = O(k\cdot n^{1/k}\log^{1-1/k} n)=\Ot(n^{1/k})$ with high probability.
\end{lemma}

\section{Technical Overview}
\paragraph{\nPSP in unweighted graphs.}
The $(\lceil \tfrac{4k}{3} \rceil - 1,\, \lceil \tfrac{4k}{3} \rceil - 1)$-approximation algorithm of Chechik~\etal~\cite{DBLP:conf/soda/ChechikHL25} works only in unweighted graphs because its analysis relies on the ability to select three specific vertices on the shortest path between $u$ and $v$, positioned approximately at distances $d(u,v)/4$, $d(u,v)/2$, and $3d(u,v)/4$ from $u$. Our first contribution is a simplification of the algorithm of~\cite{DBLP:conf/soda/ChechikHL25} in which we show that it suffices to consider only a single vertex on the shortest path that is located at roughly $d(u,v)/2$ from $u$. Moreover, we show that when $d(u,v) \equiv 0 \pmod{2}$, the additive term in the approximation can be eliminated entirely. We present this result as a warm-up in \Cref{S-4k/3-unweighted}. In \Cref{S-k-stretch-small-k}, we show that by slightly modifying the algorithm of \Cref{S-4k/3-unweighted} and using tighter analysis, one can improve the multiplicative stretch to $\ceil{\tfrac{4k}{3}-\tfrac{5}{3}}$ and remove the additive error entirely.

\paragraph{\nPSP in weighted graphs.}
In~\Cref{S-nPSP-weighted} we turn our attention to weighted graphs. 
Let $P(u,v)$ be a shortest path between $u$ and $v$, let $(u',v')$  be the heaviest edge on $P(u,v)$ and let $W_{uv}$ be the weight of $(u',v')$. 
We begin by showing that applying the algorithm from ~\Cref{S-4k/3-unweighted} to weighted graphs yields a $(\lceil \tfrac{4k}{3} \rceil - 1,\, (\lceil \tfrac{4k}{3} \rceil - 1) \cdot W_{uv})$-approximation (see \Cref{T-nPSP-weighted-Sparse-Wuv}).
The rest of~\Cref{S-nPSP-weighted} is devoted to our main result which is a $1.622k$-approximation algorithm for weighted graphs  (see \Cref{T-nPSP-Weighted-Sparse}). 

To achieve a $1.622k$-approximation, we distinguish between two cases based on the value of $W_{uv}$. When $W_{uv}$ is small, we use \Cref{T-nPSP-weighted-Sparse-Wuv} to get a good approximation. 
The interesting case is when $W_{uv}$ is large. 
Consider for example  the case that $W_{uv}=\ell(u',v')\approx d(u,v)/2$.
In this case, the $(\lceil \tfrac{4k}{3} \rceil - 1,\, (\lceil \tfrac{4k}{3} \rceil - 1) \cdot W_{uv})$-approximation algorithm gives an approximation of $\approx \tfrac{4k}{3}d(u,v)+\tfrac{4k}{3}\cdot W_{uv}\stackrel{W_{uv}\approx d(u,v)/2}\approx2k\cdot d(u,v)$, which does not break the $2k-O(1)$ approximation barrier. To handle the case that $W_{uv}$ is large, we develop a new technique called the \textit{heavy-edge} technique.

Consider again the case that $W_{uv}=\ell(u',v')\approx d(u,v)/2$. For simplicity, assume that the edge $(u',v')$ is given to us. 
In this case, we can compute an estimation $\hat{d}(u,v)$ for $d(u,v)$ as follows. 
\begin{align*}
    \hat{d}(u,v) &=\THZQuery(u,u')+\ell(u',v')+\THZQuery(v',v)\\
    &\leq (2k-1)(d(u,u')+d(v',v)) + \ell(u',v') 
    \\&\stackrel{W_{uv}\approx d(u,v)/2}\approx (2k-1)\cdot  \tfrac{d(u,v)}{2} + \tfrac{d(u,v)}{2} = k\cdot d(u,v), 
\end{align*}

which breaks the $2k-O(1)$ approximation barrier as wanted. The difficulty is that the heavy edge $(u',v')$ is not known \emph{a priori}. 
We therefore preprocess over all possible edges $(u',v')\in E$, treating each edge as a candidate heavy edge on a queried shortest path.
Specifically, we create a hash table $H(\cdot,\cdot)$ of distance estimations, and for each edge $(u',v')\in E$, $w\in B(v')$ and $i\in [k]$, we set
$$H(p_i(u'), w) = \min(H(p_i(u'), w), h_i(u') + \wt(u', v') + d(v', w))$$

Intuitively, this stores the cost of the walk $p_i(u') \rightsquigarrow u' \to v' \rightsquigarrow w$ in $H(p_i(u'), w)$.
To compute the estimation $\hat{d}(u,v)$, the algorithm considers all vertices $x\in B(u)$ and $w\in B(v)$, and updates $\hat{d}(u,v)$ as follows:
$$\hat{d}(u,v) = \min(\hat{d}(u,v), d(u,x)+H(x,w)+d(w,v))$$

For every pair $\pair{u,v}\in I$ the query takes $O(|B(u)|\cdot |B(v)|)=O(n^{2/k})$ time, since $|B(u)|,|B(v)|=O(n^{1/k})$. Therefore, the total running time for $n$ queries is $O(n^{1+2/k})$, as required.

The main technical contribution is the approximation analysis provided in \Cref{L-stretch-1.622-weighted}, which shows that this algorithm breaks the $2k-O(1)$ approximation barrier.  Specifically, we prove that this algorithm achieves a stretch of $1.622k$, regardless of the value of $W_{uv}$.
We note that the heavy-edge technique might be found useful to convert other $W_{uv}$-dependent approximation algorithms into purely multiplicative approximation algorithms that do not depend on $W_{uv}$. 

\paragraph{\tPSP in dense weighted graphs.}
In~\Cref{S-dense}, we address the \tPSP problem in dense weighted graphs. Unlike previous approaches, which construct a spanner and then run an \nPSP algorithm on top of it, our approach does not use the spanner as a black-box. We leverage a key property of the inner structure of the spanner to obtain an improved algorithm.
The $(2k-1)$-spanner of~\cite{DBLP:journals/rsa/BaswanaS07} is constructed by sparsifying the graph $G$ in $k$ iterations into a subgraph $H$. We show that if an edge $(u,v)$ is removed in iteration $i$, then the distance between $u$ and $v$ in $H$ is at most $(2i-1)d_G(u,v)$ rather than $(2k-1)d_G(u,v)$. 

Consider the shortest path $P_G(u,v)$, and let $i$ be the iteration in which the last edge from $P_G(u,v)$ is removed, i.e., after the $i$'th iteration of the spanner construction, all the edges of $P_G(u,v)$ are already removed from the graph.
If $i$ is small, then the edges of $P_G(u,v)$ are removed at an early stage, and therefore, the stretch of the entire path in $H$ is small, satisfying $d_H(u,v) \le (2i-1)d_G(u,v)$.
If $i$ is large, then there exists $(x,y)\in P_G(u,v)$ that remained until a late iteration, implying that $h_i(x) \le i\cdot \ell(x,y)$.
We use this property to get that $h_i(u)\le i\cdot d(u,v)$, and utilize the fact that $h_i(u)\le i\cdot d(u,v)$ in our algorithm, by incorporating the parameterized distance oracle of Kadria and Roditty~\cite{kadria2025fasteralgorithms2k1stretchdistance}.
See \Cref{T-nPSP-dense-k^2+2} for the complete proof.

\section{Warm-up: $(\ceil{4k/3}-1,(\ceil{4k/3}-1)\texorpdfstring{\ODD}{})$-approximation in \texorpdfstring{$\Ot$}{O}$(mn^{1/k} + n^{1+2/k})$ time}\label{S-4k/3-unweighted}

In this section, we prove the following theorem:
\begin{theorem}\label{T-nPSP-unweighted-Sparse}
    \nPSPTheorem{$G=(V, E)$ be an unweighted undirected graph}{(\ceil{\tfrac{4k}{3}}-1) d(u,v)+(\ceil{\tfrac{4k}{3}}-1)\ODD}\footnote{$x\ODD = x\cdot (d(u,v)\pmod 2)$.}
\end{theorem}
The algorithm is simple and works as follows. First, the algorithm initializes two empty hash tables $H,\hat{d}$, and for every $u\in V$ computes $B(u)$ using \Cref{L-THZ-construction}.
Then, for every $u\in V$, and every $v,w\in B(u)$ the algorithm sets $H(v,w)$ to $\min(H(v,w), d(u,v)+d(u,w))$. 
Next, for every vertex pair $\pair{u,v}\in I$, 
and for every $w\in B(u)$ and $z\in B(v)$, the algorithm sets $\hat{d}(u,v)$ to $\min(\hat{d}(u,v), d(u,w)+H(w,z)+d(z,v))$. The algorithm returns $\hat{d}(u,v)$ for every $\pair{u,v}\in I$.
A pseudocode for the algorithm is presented in \Cref{A-npsp-unweighted-4k/3}.
\begin{algorithm2e}
    \caption{$(\ceil{4k/3}-1,(\ceil{4k/3}-1)\ODD$-approximation for \nPSP}\label{A-npsp-unweighted-4k/3}
    $H \gets HashTable(); \hat{d} \gets HashTable()$ \\
    Compute $B(u)$ for every $u\in V$ [\Cref{L-THZ-construction}]\\
    \For{$u\in V$} {
        \For{$v,w\in B(u)$} {
            $H(v,w) = \min(H(v,w), d(u,v)+d(u,w))$
        }
    }
    \For{$\pair{u,v}\in I$} {
        \For{$w\in B(u)$ and $z\in B(v)$} {
            $\hat{d}(u,v) = \min(\hat{d}(u,v), d(u,w)+H(w,z)+d(z,v))$
        }
    } 
    \Return $\{\hat{d}(u,v) \mid \pair{u,v}\in I\}$
\end{algorithm2e}

Let $\pair{u,v}\in I$ be a vertex pair from the input.
To prove \Cref{T-nPSP-unweighted-Sparse} we need to prove that:
$$\hat{d}(u,v) \le (\ceil{4k/3}-1)d(u,v)+(\ceil{4k/3}-1)\ODD$$
Let $t=d(u,v)/2+0.5\ODD$. Since the graph is unweighted, there is a vertex $\tau\in P(u,v)$ such that $d(u,\tau)=t$ and $d(v,\tau)=t-1\ODD\le t$. In \Cref{L-warmup-4k/3-unweighted-correctness}, we show that in such a case  $\hat{d}(u,v) \le (\ceil{4k/3}-1)2t$. Since $t=d(u,v)/2+0.5\ODD$ this implies that
$$\hat{d}(u,v) \le (\ceil{4k/3}-1)2t = (\ceil{4k/3}-1) d(u,v)+(\ceil{4k/3}-1)\ODD,$$
as required. 
We remark that the only point at which the assumption that $G$ is unweighted is used is in setting 
$t = d(u,v)/2 + 0.5\ODD$.
Apart from this, the proof of \Cref{L-warmup-4k/3-unweighted-correctness} remains valid for weighted graphs, provided the lemma's condition is satisfied.

\begin{lemma}\label{L-warmup-4k/3-unweighted-correctness}
    Let $t$ be a value such that there exists a vertex $ \tau \in P(u,v) $ satisfying $ d(u, \tau), d(v, \tau) \le t$. Then, $\hat{d}(u,v) \le (\lceil 4k/3 \rceil - 1) \cdot 2t.$
\end{lemma}
\begin{proof}
    First, we show that $\hat{d}(u, v) \le \THZQuery(u, v)$.
    \begin{cclaim}\label{CC-hat-d-u-v-le-thzq}
        $\hat{d}(u,v) \le \THZQuery(u, v)$
    \end{cclaim}
    \begin{proof}
        The query of $\THZQuery(u, v)$ finds the minimum $i$ such that $p_i(u)\in B(v)$ or $p_i(v)\in B(u)$.
        Then it returns $h_i(u)+d(p_i(u),v)$ or $d(u,p_i(v))+h_i(v)$. Wlog assume that $p_i(v)\in B(u)$, then in our algorithm, in the final loop there is an iteration with $w=z=p_i(v)$, therefore after this iteration we have that $\hat{d}(u,v)\le d(u,p_i(v))+h_i(v)=\THZQuery(u,v)$, as required.
    \end{proof}
    
    Let $i$ be the maximum index for which $h_i(\tau)\le i\cdot t$. Such an  index must exist  since $h_0(\tau)=d(\tau,\tau)=0=0\cdot t$.    
    We divide the proof into three cases (See \Cref{fig:npsp_small_t} for an illustration of the $3$ cases):
    \begin{enumerate}
        \item $p_i(\tau)\in B(u)\cap B(v)$.\label{E-1}
        \item $p_i(\tau)\in B(v)\setminus B(u)$ or $p_i(\tau)\in B(u)\setminus B(v)$. \label{E-3}
        \item $p_i(\tau)\notin B(u)$ and $p_i(\tau)\notin B(v)$.\label{E-2}
    \end{enumerate}
    We start with case~\ref{E-1}. In this case $p_i(\tau) \in B(u) \cap B(v)$. Therefore, in the final loop of the algorithm, there is an iteration in which $w = z = p_i(\tau)$ and the algorithm 
    sets $\hat{d}(u,v)$ to $\min(\hat{d}(u,v), d(u,p_i(\tau))+H(p_i(\tau),p_i(\tau))+d(p_i(\tau),v))$. 
    Therefore:
    \begin{align*}
        \hat{d}(u,v) &\le d(u,p_i(\tau))+H(p_i(\tau),p_i(\tau))+d(p_i(\tau),v) \stactri \le d(u,\tau)+h_i(\tau)+0+h_i(\tau)+d( \tau,v) 
        \\&\stackrel{h_i(\tau)\le it}\le d(u,v)+ 2it \stackrel{d(u,v)\le 2t}\le (2i+2)t \stackrel{i\le k-1}\le 2kt \le 2(\ceil{4k/3}-1)t,
    \end{align*}
    as required.
    
    For cases~\ref{E-3} and~\ref{E-2}, we will show that $\hat{d}(u, v) \le (4i + 6)t$ and that $\min(h_{i+1}(u), h_{i+1}(v)) \le (i + 1)t$. 

    \begin{figure}[t]
    \centering
    \tikzset{every picture/.style={line width=0.75pt}} 
    \input{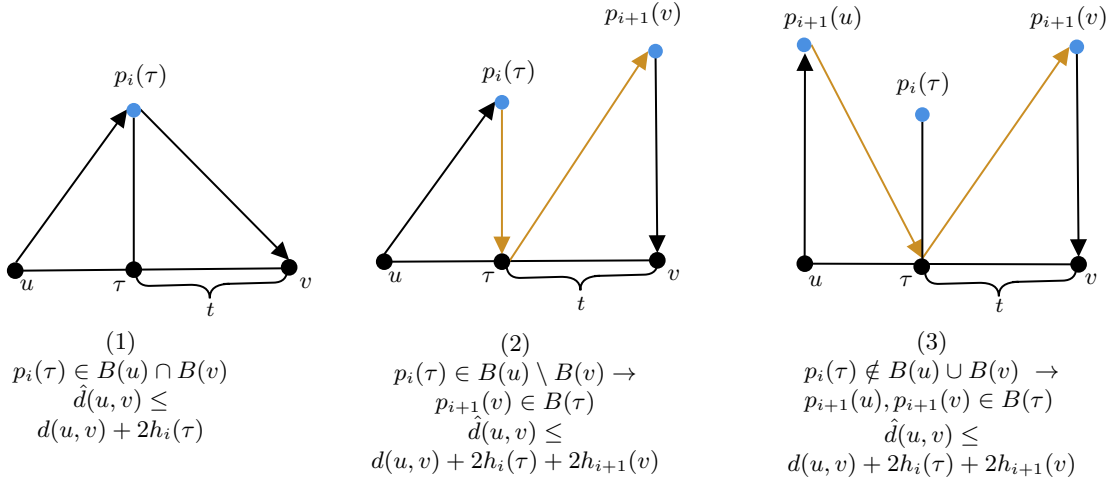}

    \caption{The three cases of \Cref{L-warmup-4k/3-unweighted-correctness}.  $i=\max\{j \mid h_j(\tau)\le j\cdot t\}$.}
    \label{fig:npsp_small_t}
\end{figure}

    From \Cref{CC-hat-d-u-v-le-thzq} we have that $\hat{d}(u, v) \le \THZQuery(u, v)$, and we can apply \Cref{L-THZ05-Q-Correctness-private} with parameter $i+1$ to obtain:
    \begin{align*}
        \hat{d}(u,v) &\le (2k-2(i+1)-1)d(u,v)+2\min(h_{i+1}(u),h_{i+1}(v)) \\&\stackrel{\min(h_{i+1}(u),h_{i+1}(v))\le (i+1)t}{\le} (2k-2i-3)d(u,v)+2(i+1)t \stackrel{d(u,v)\le2t}\le (4k-2i-4)t
    \end{align*}
    By combining this bound with the fact that $\hat{d}(u,v) \le (4i+6)t$ we get that: 
    $$\hat{d}(u,v) \le \min(4k-2i-4,4i+6)t\le (\ceil{4k/3}-1)2t,$$
    as required.

    Therefore, to complete the proof of the lemma, it remains to show that in cases~\ref{E-3} and~\ref{E-2}, we have that $\hat{d}(u,v) \le (4i+6)t$ and $\min(h_{i+1}(u),h_{i+1}(v))\le (i+1)t$.
    To prove cases~\ref{E-3} and~\ref{E-2}, we first prove the following useful claim.
    \begin{cclaim}\label{SC-pi+1inptau-0}
        If $p_i(\tau)\notin B(u)$ then $p_{i+1}(u)\in B(\tau)$ and $h_{i+1}(u)\le (i+1)t$, for every $0\le i\le k-2$.
    \end{cclaim}
    \begin{proof}
        Since $p_i(\tau)\notin B(u)$, we have that 
        $$d(u,p_{i+1}(u)) \le d(u,p_i(\tau)) \stactri\le d(u,\tau)+d(\tau,p_i(\tau)) \stackrel{h_i(\tau)\le it}\le t+it=(i+1)t.$$

        Since $i$ is the maximum index such that $h_i(\tau) \le it$, we have that $h_{i+2}(\tau) > (i+2)t$. From the triangle inequality we have: $$d(\tau, p_{i+1}(u)) \le d(\tau,u)+d(u,p_{i+1}(u)) \le t+(i+1)t=(i+2)t < h_{i+2}(\tau),$$ and therefore $p_{i+1}(u)\in B_{i+1}(\tau)$, as required.
    \end{proof}
    We are now ready to prove cases \ref{E-3} and \ref{E-2}. In case \ref{E-3} $p_i(\tau)\in B(v)\setminus B(u)$ or $p_i(\tau)\in B(u)\setminus B(v)$. Wlog, we assume that $p_i(\tau)\in B(v)\setminus B(u)$. Since $p_i(\tau)\notin B(u)$, from \Cref{SC-pi+1inptau-0} we have that $p_{i+1}(u)\in B(\tau)$ and $h_{i+1}(u)\le (i+1)t$. 
    Since $p_i(\tau),p_{i+1}(u)\in B(\tau)$ we have that $H(p_{i+1}(u), p_i(\tau)) \le d(p_{i+1}(u), \tau)+h_i(\tau)$.
    Since $p_{i+1}(u) \in B(u)$ and $p_i(\tau)\in B(v)$, in the final loop of the algorithm, we have an iteration in which $w=p_{i+1}(u)$ and $z=p_i(\tau)$. Therefore:
    \begin{align*}
        \hat{d}(u,v) &\le d(u,p_{i+1}(u))+H(p_{i+1}(u), p_i(\tau))+d(p_i(\tau),v)
        \\&\le d(u,p_{i+1}(u)) + d(p_{i+1}(u),\tau) + d(\tau,p_i(\tau)) + d(p_i(\tau),v)
        \\&\stactri\le d(u,v) + 2h_{i+1}(u) + 2h_i(\tau) \stackrel{h_{i+1}(u)\le (i+1)t}\le d(u,v)+2(i+1)t+2it \stackrel{d(u,v)\le 2t}\le (4i+4)t,
    \end{align*}
    as required.

    Finally, we complete the proof by proving case \ref{E-2}. In case \ref{E-2} we have $p_i(\tau)\notin B(u)$ and $p_i(\tau)\notin B(v)$. From \Cref{SC-pi+1inptau-0} we have that $p_{i+1}(u), p_{i+1}(v)\in B(\tau)$ and $h_{i+1}(u),h_{i+1}(v)\le (i+1)t$. Since $p_{i+1}(u), p_{i+1}(v)\in B(\tau)$ we have that $H(p_{i+1}(u),p_{i+1}(v)) \le d(p_{i+1}(u),\tau)+d(\tau,p_{i+1}(v))$. 
    Since $p_{i+1}(u)\in B(u)$ and $p_{i+1}(v)\in B(v)$, in the final loop of the algorithm we have an iteration in which $w=p_{i+1}(u)$, and $z=p_{i+1}(v)$. Therefore:
    \begin{align*}
        \hat{d}(u,v) &\le d(u,p_{i+1}(u))+H(p_{i+1}(u),p_{i+1}(v))+d(p_{i+1}(v),v)
        \\&\le d(u,p_{i+1}(u)) + d(p_{i+1}(u),\tau)+d(\tau,p_{i+1}(v)) + d(p_{i+1}(v),v) \\&\stactri\le d(u,v) + 2h_{i+1}(u) + 2h_{i+1}(v) 
        \\&\le d(u,v)+4(i+1)t\stackrel{d(u,v)\le 2t}\le (4i+6)t,
    \end{align*} 
    as required.
\end{proof}

Next, we show that the running time of the algorithm is $\tilde{O}(mn^{1/k} + n^{1+2/k})$.
\begin{lemma}\label{L-npsp-unweighted-running-time}
    The algorithm takes $\tilde{O}(mn^{1/k} + n^{1+2/k})$ time.
\end{lemma}
\begin{proof}
    Computing $B(u)$ for every $u \in V$ takes $\tilde{O}(mn^{1/k})$ time.
    Constructing $H$ requires $O\left(\sum_{u \in V} |B(u)|^2\right) = \tilde{O}(n^{1 + 2/k})$.
    Estimating the distances takes $O\left(\sum_{(u,v) \in I} |B(u)| \cdot |B(v)|\right) = \tilde{O}(n^{1 + 2/k})$. Therefore, the running time is $\tilde{O}(mn^{1/k} + n^{1+2/k})$, as required.
\end{proof}

\Cref{T-nPSP-unweighted-Sparse} follows from \Cref{L-warmup-4k/3-unweighted-correctness} and \Cref{L-npsp-unweighted-running-time}.

\section{A $1.622k$-approximation algorithm for \nPSP in weighted graphs}\label{S-nPSP-weighted}
In this section, we turn our attention to the \nPSP problem in weighted undirected graphs. 
We start by extending \Cref{T-nPSP-unweighted-Sparse} to weighted graphs by bounding the additive error with $(\ceil{4k/3}-1)W_{uv}$ instead of $(\ceil{4k/3}-1)$ while keeping the multiplicative error the same.\footnote{Recall that $W_{uv}=\max_{(w,z)\in P(u,v)}\wt(w,z)$.}

\Reminder{T-nPSP-weighted-Sparse-Wuv}

\begin{proof}
    To prove the theorem we analyze \Cref{A-npsp-unweighted-4k/3} when the input graph is weighted  and show that $d(u,v) \le \hat{d}(u,v) \le(\ceil{\tfrac{4k}{3}}-1) d(u,v)+(\ceil{\tfrac{4k}{3}}-1)W_{uv}$.  
    Let \\$t=\min_{\tau\in P(u,v)}(\max(d(u,\tau),d(v,\tau)))$. Notice that in unweighted graphs \\$\min_{\tau\in P(u,v)}(\max(d(u,\tau),d(v,\tau)))=d(u,v)/2+0.5\ODD$.
    This generalized definition of $t$ ensures that there exists $\tau \in P(u, v)$ such that $d(u, \tau), d(v, \tau) \le t$, and therefore we can apply \Cref{L-warmup-4k/3-unweighted-correctness} to get that $\hat{d}(u,v) \le (\ceil{\tfrac{4k}{3}}-1)2t$.

    We now bound $t$ in terms of $d(u, v)$ and $W_{uv}$. Consider an edge $(\tau_u, \tau_v)$ along the shortest path $P(u, v)$ such that $d(u, \tau_u) \le d(u, v)/2$ and $d(u, \tau_v) > d(u, v)/2$. (See \Cref{fig:tau_u_tau_v} for an illustration.) Since $d(u, \tau_v) + d(v, \tau_u) = d(u, v) + \wt(\tau_u, \tau_v)$, it follows that: 
    $$t\le \min(d(u,\tau_v),d(v,\tau_u))\le 
    (d(u, \tau_v) + d(v, \tau_u))/2
    \le (d(u,v)+\wt(\tau_u,\tau_v))/2 \le d(u,v)/2+W_{uv}/2,$$
     therefore $t\le d(u,v)/2+W_{uv}/2$ and we get that: 
    $$\hat{d}(u,v)\le (\ceil{\tfrac{4k}{3}}-1)2t \le (\ceil{\tfrac{4k}{3}}-1)d(u,v)+(\ceil{\tfrac{4k}{3}}-1)W_{uv},$$
    as required. From~\Cref{L-npsp-unweighted-running-time} it follows that the running time is $\Ot(mn^{1/k}+n^{1+2/k})$.
    \begin{figure}[t]
    \centering
    \tikzset{every picture/.style={line width=0.75pt}} 
    \tikzset{every picture/.style={line width=0.75pt}} 

\begin{tikzpicture}[x=0.75pt,y=0.75pt,yscale=-1,xscale=1]

\draw    (97.67,111) -- (272.67,110) ;
\draw  [fill={rgb, 255:red, 0; green, 0; blue, 0 }  ,fill opacity=1 ] (92.83,111) .. controls (92.83,108.33) and (95,106.17) .. (97.67,106.17) .. controls (100.34,106.17) and (102.5,108.33) .. (102.5,111) .. controls (102.5,113.67) and (100.34,115.83) .. (97.67,115.83) .. controls (95,115.83) and (92.83,113.67) .. (92.83,111) -- cycle ;
\draw  [fill={rgb, 255:red, 0; green, 0; blue, 0 }  ,fill opacity=1 ] (267.83,110) .. controls (267.83,107.33) and (270,105.17) .. (272.67,105.17) .. controls (275.34,105.17) and (277.5,107.33) .. (277.5,110) .. controls (277.5,112.67) and (275.34,114.83) .. (272.67,114.83) .. controls (270,114.83) and (267.83,112.67) .. (267.83,110) -- cycle ;
\draw  [fill={rgb, 255:red, 0; green, 0; blue, 0 }  ,fill opacity=1 ] (157.83,111) .. controls (157.83,108.33) and (160,106.17) .. (162.67,106.17) .. controls (165.34,106.17) and (167.5,108.33) .. (167.5,111) .. controls (167.5,113.67) and (165.34,115.83) .. (162.67,115.83) .. controls (160,115.83) and (157.83,113.67) .. (157.83,111) -- cycle ;
\draw  [fill={rgb, 255:red, 0; green, 0; blue, 0 }  ,fill opacity=1 ] (200.83,111) .. controls (200.83,108.33) and (203,106.17) .. (205.67,106.17) .. controls (208.34,106.17) and (210.5,108.33) .. (210.5,111) .. controls (210.5,113.67) and (208.34,115.83) .. (205.67,115.83) .. controls (203,115.83) and (200.83,113.67) .. (200.83,111) -- cycle ;
\draw   (205.67,109) .. controls (205.62,104.33) and (203.27,102.02) .. (198.6,102.07) -- (162.1,102.41) .. controls (155.43,102.47) and (152.08,100.17) .. (152.03,95.5) .. controls (152.08,100.17) and (148.77,102.53) .. (142.1,102.6)(145.1,102.57) -- (105.6,102.94) .. controls (100.93,102.99) and (98.62,105.34) .. (98.67,110.01) ;
\draw   (161.67,111) .. controls (161.71,115.67) and (164.06,117.98) .. (168.73,117.94) -- (207.23,117.59) .. controls (213.9,117.53) and (217.25,119.83) .. (217.29,124.5) .. controls (217.25,119.83) and (220.56,117.47) .. (227.23,117.41)(224.23,117.43) -- (265.73,117.06) .. controls (270.4,117.01) and (272.71,114.66) .. (272.67,109.99) ;

\draw (99.67,114) node [anchor=north west][inner sep=0.75pt]  [font=\small] [align=left] {$\displaystyle u$};
\draw (274.67,113) node [anchor=north west][inner sep=0.75pt]  [font=\small] [align=left] {$\displaystyle v$};
\draw (146.67,113) node [anchor=north west][inner sep=0.75pt]  [font=\small] [align=left] {$\displaystyle \tau _{u}$};
\draw (207.67,98) node [anchor=north west][inner sep=0.75pt]  [font=\small] [align=left] {$\displaystyle \tau _{v}$};
\draw (61,164) node [anchor=north west][inner sep=0.75pt]   [align=left] {$ $};
\draw (131,80) node [anchor=north west][inner sep=0.75pt]  [font=\footnotesize] [align=left] {$\displaystyle d( u,\tau _{v})$};
\draw (197,125) node [anchor=north west][inner sep=0.75pt]  [font=\footnotesize] [align=left] {$\displaystyle t=d( v,\tau _{u})$};

\end{tikzpicture}
    \caption{Illustration of $P(u,v)$, where $(\tau_u,\tau_v)\in P(u,v)$ and $d(u,\tau_u)\le d(u,v)/2$, $d(u,\tau_v)>d(u,v)/2$.}\label{fig:tau_u_tau_v}
\end{figure}
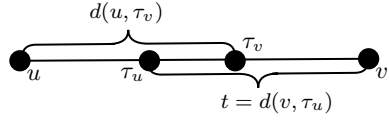

\end{proof}

Next, we present the main result of this paper: the first truly $(2-\alpha)k$, for $\alpha>0$, approximation \nPSP algorithm for weighted undirected graphs with non-negative real edge weights. 
We prove:

\Reminder{T-nPSP-Weighted-Sparse}

The algorithm works as follows. First, similar to \Cref{S-4k/3-unweighted}, the algorithm initializes two empty hash tables $H,\hat{d}$, and for every $u\in V$ computes $B(u)$ using \Cref{L-THZ-construction}.
For every $u\in V$, and for every $v,w\in B(u)$ the algorithm sets $H(v,w)$ to $\min(H(v,w), d(u,v)+d(u,w))$. 

The new insight comes from the following step (the heavy-edge technique): For every edge $(u,v)\in E$ and for every $i\in [k]$, the algorithm sets $H(p_i(u),w)$ to $\min(H(p_i(u),w), h_i(u)+\ell(u,v)+d(v,w))$, for every $w\in B(v)$.

Next, the algorithm approximates the distances for every pair $\pair{u,v}\in I$, the same as \Cref{S-4k/3-unweighted}:
For every $w\in B(u)$ and $z\in B(v)$, the algorithm sets $\hat{d}(u,v)$ to $\min(\hat{d}(u,v), d(u,w)+H(w,z)+d(z,v))$.
The algorithm returns $\hat{d}(u,v)$ for every $\pair{u,v}\in I$.
A pseudocode for the algorithm is presented in \Cref{A-npsp-weighted-1.622}.
\begin{algorithm2e}
    \caption{$1.622k$-approximation for \nPSP}\label{A-npsp-weighted-1.622}
    $H \gets HashTable(); \hat{d}(u,v) \gets HashTable()$ \\
    Compute $B(u)$ for every $u\in V$ [\Cref{L-THZ-construction}]\\
    \For{$u\in V$} {
        \For{$v,w\in B(u)$} {
            $H(v,w) = \min(H(v,w), d(u,v)+d(u,w))$}
    }
    \For{$(u,v)\in E$} {
        \For{$i\in [k]$ and $w\in B(v)$} {
            $H(p_i(u),w) = \min(H(p_i(u),w), h_i(u)+\ell(u,v)+d(v,w))$}
    }
    \For{$\pair{u,v}\in I$} {
        \For{$w\in B(u)$ and $z\in B(v)$} {
            $\hat{d}(u,v) = \min(\hat{d}(u,v), d(u,w)+H(w,z)+d(z,v))$
        }
    } 
    \Return $\{\hat{d}(u,v) \mid \pair{u,v}\in I \}$
\end{algorithm2e}

We note that \Cref{A-npsp-weighted-1.622} differs from \Cref{A-npsp-unweighted-4k/3} in the loop on lines 6--8, which iterates over every $(u,v) \in E$, $i \in [k]$ and $w \in B(v)$ and sets $H(p_i(u),w)$ to $\min(H(p_i(u),w), h_i(u)+\ell(u,v)+d(v,w))$. Notice that this adaptation takes $O(\sum_{(u,v)\in E} k \cdot |B(v)|) = \Ot( m \cdot n^{1/k})$, and therefore it does not increase our running time.
This modification helps us bound $\hat{d}(u,v)$ in the case where $W_{uv}$ is large (hence we call it the heavy-edge technique), and obtain an approximation that does not depend on $W_{uv}$, as shown in the following lemma. Let $\pair{u,v}\in I$ be a vertex pair from the input.
\begin{lemma}\label{L-stretch-1.622-weighted}
    $\hat{d}(u,v) \le 1.622k d(u,v)$
\end{lemma}

\begin{proof}
    Let $t=\min_{w\in P(u,v)}(\max(d(u,w),d(w,v)))$. This implies that $d(u,v)/2 \le t \le d(u,v)$. Let $\tau\in P(u,v)$ be the vertex, for which $\max(d(u,\tau), d(\tau,v))=t$.
    The proof is divided into two cases, according to the value of $t$. The case that $t \le (1/2+c)d(u,v)$, and the case that $t > (1/2+c)d(u,v)$, for some constant $0<c<1/2$, to be determined later. 
   
    At a high level, if $t \le (1/2 + c)\cdot d(u,v)$, then there exists a vertex on the path $P(u,v)$ that lies close to its midpoint. 
    In this situation, we can directly apply the bound established for the unweighted case. 
    On the other hand, if $t > (1/2+c)d(u,v)$, then the next edge $(\tau,\tau_v) \in P(u,v)$ must have relatively large weight, in this case we use the heavy-edge technique, which is the main technical contribution of this lemma. 
    Since $(\tau,\tau_v)$ is considered in the algorithm explicitly, the approximation for its length is $1$ (exact weight), since it is a heavy edge in the path, we get an overall less than $2k$-approximation for the entire path.
   
    In the proof, we have two bounds according to the value of $c$, and then we choose the optimal $c$ that minimizes the maximum of these bounds.
    Formally, the analysis yields two different upper bounds on the approximation, each corresponding to one of the two cases and parameterized by $c$. The final step of the proof selects the value of $c$ that minimizes the worst case of these two bounds, thereby optimizing the overall guarantee.
    
    Consider first the case that  $t \le (1/2+c)d(u,v)$. Since $d(u,\tau), d(v,\tau)\le t$ it follows from \Cref{L-warmup-4k/3-unweighted-correctness} that: $$\hat{d}(u,v)\le (\ceil{\tfrac{4k}{3}}-1)\cdot 2t \le (4k/3)(1+2c)d(u,v)$$
    
    Next, we consider the case where $t > (1/2+c)d(u,v)$.
    Assume, wlog, that $t=d(\tau,v) > d(u,\tau)$. 
    Since $d(u,\tau)+d(\tau,v)=d(u,v)$, and $t \ge (1/2+c)d(u,v)$ we get that $d(u,\tau) = d(u,v) - t \le (1/2-c)d(u,v)$.
    Let $\tau_v\in P(\tau,v)$ be the first vertex from $\tau$ to $v$, i.e., $d(v,\tau)=d(v,\tau_v)+\ell(\tau_v,\tau)$. From the minimality of $t$, we know that $d(u,\tau_v) \ge t \ge (1/2+c)d(u,v)$, and therefore $d(v,\tau_v) = d(u,v)-d(u,\tau_v) \le d(u,v)-t\le (1/2-c)d(u,v)$. We conclude that 
    \begin{equation}\label{e-u-tau-v-tau_v}
        d(u,\tau), d(v,\tau_v) \le (1/2-c)d(u,v)
    \end{equation}
    
    Let $j_u$ be the first index such that $p_{j_u}(\tau)\in B(u)$. Since $p_j(\tau)\notin B(u)$ for every $j<j_u$, we can apply \Cref{L-Bound-delta-2d} to obtain:
    \begin{equation}\label{e-bound-h-ju}
        h_{j_u}(\tau) \le {j_u}\cdot 2d(u,\tau) \stackrel{\ref{e-u-tau-v-tau_v}}\le (1-2c)j_u d(u,v)
    \end{equation}
    
    Let $j_v$ be the first index such that $p_{j_v}(\tau_v)\in B(v)$ or $p_{j_v}(v)\in B(\tau_v)$. Since $p_j(\tau_v)\notin B(v)$ and $p_j(v)\notin B(\tau_v)$ for every $j<j_v$, we can apply \Cref{L-Bound-p-i-d} to obtain:
    \begin{equation}\label{e-bound-h-jv}
        h_{j_v}(\tau_v),h_{j_v}(v) \le j_v d(\tau_v,v) \stackrel{\ref{e-u-tau-v-tau_v}}\le (1/2-c)j_v d(u,v)
    \end{equation}
    
    Next, we prove  the following three bounds on $\hat{d}(u,v)$:\begin{figure}[t]
    \centering
    \tikzset{every picture/.style={line width=0.75pt}} 
    \input{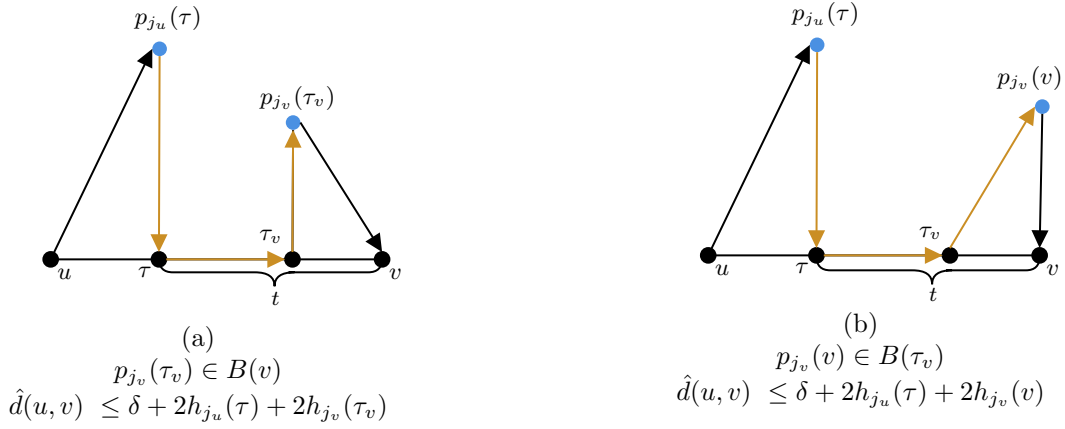}

    \caption{The two cases of \Cref{C-eq:1}. $t \ge (1/2+c)\delta$. $j_u=\min(i\mid p_i(\tau)\in B(u))$,\\ $j_v=\min(i\mid p_i(v)\in B(\tau_v) \vee p_i(\tau_v)\in B(v))$.}
    \label{fig:npsp_big_t}
\end{figure}

    \begin{itemize}
        \item[] \ref{C-eq:1}. $\hat{d}(u,v) \le (1+2(1-2c)j_u+(1-2c)j_v)d(u,v)$ 
        \item[] \ref{C-eq:2}. $\hat{d}(u,v) \le (2k-j_v-1-2cj_v)d(u,v)$ 
        \item[] \ref{C-eq:3}. $\hat{d}(u,v) \le (2k-4j_uc+2c-3)d(u,v)$ 
    \end{itemize}
    
    \begin{cclaim}\label{C-eq:1}
        $\hat{d}(u,v) \le d(u,v)+2(1-2c)j_u d(u,v)+(1-2c)j_v d(u,v)$
    \end{cclaim}
    \begin{proof}
    We divide the proof into two cases. The case that $p_{j_v}(v)\in B(\tau_v)$, and the case that $p_{j_v}(\tau_v)\in B(v)$. (See \Cref{fig:npsp_big_t} for an illustration.)

    Consider first, the case that $p_{j_v}(v) \in B(\tau_v)$. Since $p_{j_v}(v) \in B(\tau_v)$, after the edge $(\tau, \tau_v)$ is considered in the loop of lines 6--8, with $i = j_u$ and $w = p_{j_v}(v)$, it is guaranteed that:
    \begin{equation}\label{eq:bound-h-pjv-v}
        H(p_{j_u}(\tau), p_{j_v}(v)) \le h_{j_u}(\tau)+\ell(\tau,\tau_v)+d(\tau_v,p_{j_v}(v))
    \end{equation}
    Since $p_{j_u}(\tau)\in B(u)$ and $p_{j_v}(v)\in B(v)$, in the final loop of the algorithm there is an iteration in which $w=p_{j_u}(\tau)$ and $z=p_{j_v}(v)$.  Therefore:
    \begin{align*}
        \hat{d}(u,v) &\le d(u,p_{j_u}(\tau))+H(p_{j_u}(\tau), p_{j_v}(v))+h_{j_v}(v) \\&\stackrel{\ref{eq:bound-h-pjv-v}}\le d(u,p_{j_u}(\tau))+(h_{j_u}(\tau)+\ell(\tau,\tau_v)+d(\tau_v,p_{j_v}(v))) + h_{j_v}(v) 
        \\& \stactri \le d(u,\tau)+h_{j_u}(\tau)+h_{j_u}(\tau)+\ell(\tau,\tau_v)+d(\tau_v,v)+h_{j_v}(v)+h_{j_v}(v)
        \\&=d(u,v)+2h_{j_u}(\tau)+2h_{j_v}(v) 
        \stackrel{\ref{e-bound-h-ju},\ref{e-bound-h-jv}}\le d(u,v)+2(1-2c)j_u d(u,v)+(1-2c)j_v d(u,v),
    \end{align*}       
    as required.
    We consider now the case that $p_{j_v}(\tau_v)\in B(v)$. Since 
    $p_{j_v}(\tau_v)\in B(\tau_v)$, after the edge $(\tau, \tau_v)$ is considered in the loop of lines 6--8, with $i = j_u$ and $w = p_{j_v}(\tau_v)$, we get:
    \begin{equation}\label{eq:bound-h-pjtauv-v}
        H(p_{j_u}(\tau), p_{j_v}(\tau_v)) \le h_{j_u}(\tau)+\ell(\tau,\tau_v)+h_{j_v}(\tau_v)
    \end{equation}
    
    Since $p_{j_u}(\tau)\in B(u)$ and $p_{j_v}(\tau_v)\in B(v)$, in the final loop of the algorithm there is an iteration in which $w=p_{j_u}(\tau)$ and $z=p_{j_v}(\tau_v)$.  Therefore:
    \begin{align*}
        \hat{d}(u,v) &\le d(u,p_{j_u}(\tau))+H(p_{j_u}(\tau),p_{j_v}(\tau_v))+d(p_{j_v}(\tau_v),v) 
        \\&\stackrel{\ref{eq:bound-h-pjtauv-v}}\le d(u,p_{j_u}(\tau)) + (h_{j_u}(\tau)+\ell(\tau,\tau_v)+h_{j_v}(\tau_v)) + d(p_{j_v}(\tau_v),v)
        \\&\stactri \le d(u,\tau)+h_{j_u}(\tau)+h_{j_u}(\tau)+\ell(\tau,\tau_v)+h_{j_v}(\tau_v) + h_{j_v}(\tau_v) + d(\tau_v,v) 
        \\&= d(u,v)+2h_{j_u}(\tau)+2h_{j_v}(\tau_v) 
        \stackrel{\ref{e-bound-h-ju},\ref{e-bound-h-jv}}\le d(u,v)+2(1-2c)j_u d(u,v)+(1-2c)j_v d(u,v),
    \end{align*}
    as required.
    \end{proof}
    \begin{cclaim}\label{C-eq:2}
        $\hat{d}(u,v) \le (2k-j_v-1-2cj_v)d(u,v)$
    \end{cclaim}
    \begin{proof}
        From \Cref{CC-hat-d-u-v-le-thzq} we have that $\hat{d}(u,v) \le \THZQuery(u,v)$. Therefore, applying \Cref{L-THZ05-Q-Correctness-private} with $j_v$, we get that:
        \begin{align*}
        \hat{d}(u,v) &\le (2k-2j_v-1)d(u,v) + 2h_{j_v}(v) \stackrel{\ref{e-bound-h-jv}}\le (2k-2j_v-1)d(u,v) + (1-2c)j_vd(u,v)
        \\&=(2k-j_v-1-2cj_v)d(u,v),
        \end{align*}
        as required.
    \end{proof}
    \begin{cclaim}\label{C-eq:3}
        $\hat{d}(u,v) \le (2k-4j_uc+2c-3)d(u,v)$
    \end{cclaim}
    \begin{proof}
        From \Cref{L-Bound-delta-2d} we have that 
        \begin{equation}\label{eqq:7}
            h_{j_u-1}(\tau) \le (j_u-1)\cdot 2d(u,\tau) 
        \end{equation}
        From the minimality of $j_u$ it follows that  $p_{j_u-1}(\tau)\notin B(u)$. Thus: 
        \begin{equation}\label{e-h-j-u-u-special}
            h_{j_u}(u) \le d(u, p_{j_u-1}(\tau)) \stactri \le d(u,\tau)+h_{j_u-1}(\tau) \stackrel{\ref{eqq:7}} \le d(u,\tau) + (2j_u-1)d(u,\tau) \stackrel{\ref{e-u-tau-v-tau_v}}\le (1/2-c)(2j_u-1)d(u,v)
        \end{equation}

        From \Cref{CC-hat-d-u-v-le-thzq} we have that $\hat{d}(u,v) \le \THZQuery(u,v)$. Therefore, from \Cref{L-THZ05-Q-Correctness-private} with parameter $j_u$, we get that 
        \begin{align*}
            \hat{d}(u,v) &\le (2k-2j_u-1)d(u,v) + 2h_{j_u}(u) \stackrel{\ref{e-h-j-u-u-special}}\le (2k-2j_u-1)d(u,v) + 2(1/2-c)(2j_u-1)d(u,v) 
            \\&= (2k-4j_uc+2c-3)d(u,v), \text{ as required.}
        \end{align*}
    \end{proof}
    
    Using the three bounds on $\hat{d}(u,v)$ we get that:
    $$\hat{d}(u,v) \le \min\big( 2k-j_v-1-2cj_v, 2k-4j_uc+2c-3, 1+2(1-2c)j_u+(1-2c)j_v\big)\cdot d(u,v)$$
    From 
     \href{https://www.wolframalpha.com/input?i=2k-y-1-2cy+%3D+2k-4xc%2B2c-3%2C++2k-4xc%2B2c-3+%3D+1%2B2%281-2c%29x%2B%281-2c%29y}{\footnotemark}
     \footnotetext{https://www.wolframalpha.com/input?i=2k-y-1-2cy+\%3D+2k-4xc\%2B2c-3\%2C++2k-4xc\%2B2c-3+\%3D+1\%2B2\%281-2c\%29x\%2B\%281-2c\%29y, where $j_u=x$ and $j_v=y$} we have that:
    $$\hat{d}(u,v) \le (2k-4c\frac{k+2ck-3}{1+4c-4c^2}+2c-3)\cdot d(u,v)$$

    Therefore, we obtained that if $t\ge (1/2+c)d(u,v)$ then $\hat{d}(u,v) \le (2k-4c\frac{k+2ck-3}{1+4c-4c^2}+2c-3)\cdot d(u,v)$, and if $t \le (1/2+c)d(u,v)$ then $\hat{d}(u,v)\le (4k/3)(1+2c)d(u,v)$. 
    We want to choose the optimal value of $c$ such that the maximum of these two bounds is minimal. Since when $c$ increases, one term increases and the other decreases, we want to have that 
    \begin{align*}
        (2k-4c\frac{k+2ck-3}{1+4c-4c^2}+2c-3) &= (4k/3)(1+2c) \\
        2  - \frac{4 c (2 c + 1)}{1 + 4 c -4 c^2} &= (4/3)(1+2c) \\
        \frac{16 c^3 - 32 c^2 - 6 c + 1}{4 c^2 - 4 c - 1} &= 0
    \end{align*}
    From \href{https://www.wolframalpha.com/input?i=%281%2B2x%29%5Cfrac%7B4k%7D%7B3%7D+%3D+2k-4x%5Cfrac%7Bk%2B2xk%7D%7B1%2B4x-4x%5E2%7D}{\footnotemark}
     \footnotetext{https://www.wolframalpha.com/input?i=\%281\%2B2x\%29\%5Cfrac\%7B4k\%7D\%7B3\%7D+\%3D+2k-4x\%5Cfrac\%7Bk\%2B2xk\%7D\%7B1\%2B4x-4x\%5E2\%7D,  where $x=c$}
    we have that $c\approx 0.107912$ is the solution for this example, and therefore we have that:
    $$\hat{d}(u,v) \le \frac{4k}{3}d(u,v)(1+2c) \le (1+2\cdot 0.108)\frac{4k}{3}d(u,v) < 1.622kd(u,v),$$
    as required. \footnote{Notice that for $c\approx 0.107912$ we have that $2c-3$ is negative, and therefore removing this term from the first line does not harm the upper bound.}
\end{proof}

Next, we bound the running time of the algorithm.
\begin{lemma}\label{L-1.622-runtime}
    The running time of the algorithm is $\Ot(mn^{1/k}+n^{1+2/k})$
\end{lemma}
\begin{proof}
    Computing $B(u)$ for every $u \in V$ takes $\tilde{O}(mn^{1/k})$ time.
    Constructing $H$ requires \\ $O\left(\sum_{u \in V} |B(u)|^2 + \sum_{(u,v)\in E}|\{p_i(u) \mid i\in [k]\}\times B(v)|\right) = \tilde{O}(n^{1 + 2/k}+mn^{1/k})$.
    Estimating the distances takes $O\left(\sum_{(u,v) \in I} |B(u)| \cdot |B(v)|\right) = \tilde{O}(n^{1 + 2/k})$. Therefore, the running time is $\tilde{O}(mn^{1/k} + n^{1+2/k})$, as required.
\end{proof}

\Cref{T-nPSP-Weighted-Sparse} follows from \Cref{L-stretch-1.622-weighted,L-1.622-runtime}.
\ifC
\bibliography{bibliography}
\appendix
\else
\fi
\section{$(k^2+2k)$-approximation for \tPSP in $O(m+n^{1+2/k}+t)$ time}\label{S-dense}
In this section, we consider the \tPSP problem in dense weighted graphs, and present a near linear time algorithm for graphs with $m=\Omega(n^{1+2/k})$ edges that has a $(k^2+2k)$-approximation. We prove:

\Reminder{T-nPSP-dense-k^2+2}

Our algorithm relies on the following fast $(2k-1)$-emulator construction, which is a simplified variant of the $(2k - 1)$-spanner construction of Baswana and Sen~\cite{DBLP:journals/rsa/BaswanaS07}.
A graph $H$ is called an $\alpha$-emulator of $G$ if  
$d_G(u, v) \le d_H(u, v) \le \alpha \cdot d_G(u, v)$, for every pair of vertices $u, v \in V$.

\begin{lemma}\label{L-weak-emulator-construction}
    Let $G$ be a weighted graph. For any integer $k \ge 1$, a $(2k - 1)$-emulator with $\Ot(n^{1 + 1/k})$ edges can be constructed in $\Ot(m)$ time.
\end{lemma}
\begin{proof}
    First, the algorithm computes a vertex hierarchy $V=A_0,A_1,\dots,A_k=\emptyset$, such that $A_i$ contains every vertex from $A_{i-1}$ with probability $n^{-1/k}$, for every $0 < i < k$. Then, the algorithm computes $d(A_i,u)$, adds $(u,p_i(u))$ with weight $d(u,p_i(u))$ to $H$, for every $0<i<k$, and $u\in V$ using $k$ calls to $\Dijkstra$'s algorithm.

    Then, for every edge $(u,v)\in E$, let $i_{(u,v)}$ be the first index such that $d(u,p_{i_{(u,v)}+1}(u)) > (i_{(u,v)}+1)\cdot \ell(u,v)$, and add the edge $(p_{i_{(u,v)}}(u),v)$ with weight $d(u,p_{i_{(u,v)}}(u))+\ell(u,v)$.

    It is straightforward to see that the running time is $\Ot(m)$, and the proof that $|H|=O(n^{1+1/k})$ follows from \cite{DBLP:journals/rsa/BaswanaS07}.

    Next, we move to bound the approximation of the emulator.
    \begin{cclaim}\label{Cd-specific-corr}
        Let $(u,v)\in E$, it holds that $d_H(u,v) \le (2i_{(u,v)}+1)\ell(u,v)$.
    \end{cclaim}
    \begin{proof}
        By the definition of $i_{(u,v)}$ we have that $d_G(u,p_{i_{(u,v)}}(u))\le i\ell(u,v)$.
        Since the edge $(p_i(u),v)$ with weight $d_G(u,p_{i_{(u,v)}}(u))+\ell(u,v)$ is added to $H$, together with the edge $(u,p_i(u))$ of weight $d(u,p_i(u))$ we get that:
        $$d_H(u,v) \stactri\le d_G(u,p_{i_{(u,v)}}(u)) + d_G(u,p_{i_{(u,v)}}(u))+\ell(u,v) \le (2i+1)\ell(u,v),$$
        as required.
    \end{proof}
    
    Since $i_{(u,v)}\le k-1$ for every edge $(u,v)\in E$, we get from \Cref{Cd-specific-corr} that $d_H(u,v) \le (2k-1)\ell(u,v)$, for every edge $(u,v)\in E$, as required.
\end{proof}
We let $\Emulator(G,k)$ be the emulator of \Cref{L-weak-emulator-construction} on the graph $G$ with parameter $k$.
In addition, our algorithm uses the parameterized distance oracle of Kadria and Roditty~\cite{kadria2025fasteralgorithms2k1stretchdistance}, which we denote with $\ADO_P(G,k, S)$.
\begin{lemma}[\cite{kadria2025fasteralgorithms2k1stretchdistance}] \label{L-ADO-G_S}
    Let $k\ge 1$ and let $S\subseteq V$. There is an $O(n|S|^{\frac{1}{k}})$-space distance oracle that given two vertices $u,v\in V$, returns in $O(k)$ time an estimation $\ADO_P(G,k,S)\Query(u,v)$ such that $$d(u,v)\le \ADO_P(G,k,S)\Query(u,v) \le 2\min(d(u,S),d(v,S))+(2k-1)d(u,v)$$ The distance oracle is constructed in $O(m|S|^{\frac{1}{k}})$ time.
\end{lemma}

The \tPSP algorithm works as follows. 
Let $H$ be a $(2k-1)$-emulator of $G$ constructed with \Cref{L-weak-emulator-construction}.
For every $1\le i \le k-1$, let $A_i$ be the set used in the construction of $H$. For every $i$, we construct $\ADO_P(H,k-i,A_i)$. 
Then for every $\pair{u,v}\in I$ the algorithm returns $\min_i(\ADO_P(H,k'_i,A_i)\Query(u,v))$ as the distance estimation.
A pseudo-code for our algorithm is presented in~\Cref{A-Fast-APSP-2}.

\begin{algorithm2e}
\caption{\tPSP($G,k,I$)}\label{A-Fast-APSP-2}
$H \gets \Emulator(G, k)$ [\Cref{L-weak-emulator-construction}] \\
\For{$i\gets 1$ to $k-1$} {
    Construct $\ADO_P(H, k-i, A_i)$ [\Cref{L-ADO-G_S}]
}
\Return $\{\min_{i\in[k]}(\ADO_P(H,k-i,A_i)\Query(u,v)) \mid \pair{u,v}\in I\}$
\end{algorithm2e}

In the next lemma, we bound the value $\hat{d}(u,v)$ computed by \tPSP($G,k,I$), for any pair $\pair{u,v}\in I$.
\begin{lemma}\label{L-Stretch-APSP-Dense-k^2}
    $\hat{d}(u,v) \le (k^2+2k)\cdot d(u,v)$
\end{lemma}
\begin{proof}
For every edge $(x,y)\in E$, we denote with $i_{(x,y)}$   the first index such that $d(x,p_{i_{(x,y)}+1}(x)) > (i_{(x,y)}+1)\cdot \ell(x,y)$.
Let $i$ be $\max_{(x,y)\in P(u,v)}(i_{(x,y)})$. Let $(w,z)\in P(u,v)$ be an edge such that $i_{(w,z)}=i$. 
Therefore, we have that $d(w, p_i(w)) \le i\cdot \ell(w,z)$. 

Wlog, assume that the edge $(w,z)$ is the first edge in $P(w,v)$, therefore we have that $\ell(w,z)\le d(w,v)$.
Since $d(w, p_i(w)) \le i\cdot \ell(w,z)$, we get that 
\begin{align}\label{eqd:bound-pi}
\begin{split}
    d(u,p_i(u)) &\le d(u,p_i(w)) \stactri\le d(u,w)+d(w, p_i(w)) \le d(u,w)+i\ell(w,z) 
    \\&\stackrel{\ell(w,z)\le d(w,v)}\le d(u,w)+d(w,v)+(i-1)d(w,v)\le i\cdot d(u,v),
\end{split}
\end{align}
where the last inequality follows from the fact that $d(u,w)+d(w,v)=d(u,v)$.

Since $i \ge i_{(x,y)}$ for every $(x,y)\in P(u,v)$, from \Cref{Cd-specific-corr} we have that $d_H(x,y) \le (2i+1)\ell(x,y)$, for every $(x,y)\in P(u,v)$. Therefore, for the entire path $P(u,v)$ we have that 
\begin{equation}\label{eqd:bound-d-u-v}
    d_H(u,v) \le (2i+1)d(u,v)
\end{equation}

In the algorithm, we have $\hat{d}(u,v) \le \ADO_P(H, k-i, A_i)\Query(u,v)$.
From \Cref{L-ADO-G_S}, we have that $\ADO_P(H, k-i, A_i)\Query(u,v) \le 2\min(d(u,A_i),d(v,A_i)) + 2(k-i)d_H(u,v)$. Thus:
\begin{align*}
    \hat{d}(u,v) &\le 2\min(d(u,p_i(u)),d(v,p_i(v))) + 2(k-i)d_H(u,v) \\&\stackrel{\ref{eqd:bound-pi},\ref{eqd:bound-d-u-v}}\le 2\cdot id(u,v) + 2(k-i)\cdot (2i+1)d(u,v) 
    \\&= (2(k-i)(2i+1)+2i)d(u,v) = (4ik+2k-4i^2)d(u,v) \le (k^2+2k)d(u,v),
\end{align*}
where the last inequality follows from the following discussion.
Let $f(i)=4ik+2k-4i^2$. We have that $f'(i)=4k-8i$, and since $f''(i)=-8$ is negative, we have that this function has a maximum when $f'(i)=0$, i.e., the maximum of the function is when $i=k/2$, therefore $f(i) \le f(k/2)=k^2+2k$, as required. 
\end{proof}

Next, in the following lemma, we bound the running time of \Cref{A-Fast-APSP-2}.
\begin{lemma}\label{L-runtime-APSP-Dense-k^2}
    The running time of \tPSP is $\Ot(m+n^{1+2/k})$
\end{lemma}
\begin{proof}
    From \Cref{L-weak-emulator-construction} constructing $\Emulator(G, k)$ takes $\Ot(m)$ time.
    From \Cref{L-ADO-G_S} constructing $\ADO_P(H, k'_i, A_i)$ takes $$\Ot(|H||A_i|^{k'_i})=\Ot(n^{1+1/k}\cdot (n^{\frac{k-i}{k}})^{\frac{1}{k-i}})=\Ot(n^{1+1/k+\frac{k-i}{k\cdot (k-i)}})=\Ot(n^{1+2/k})$$

    Overall, we get that the running time is $\Ot(m+ k\cdot n^{1+2/k})=\Ot(m+n^{1+2/k})$, as required.
\end{proof}

\Cref{T-nPSP-dense-k^2+2} follows from \Cref{L-Stretch-APSP-Dense-k^2,L-runtime-APSP-Dense-k^2}.

\ifC
\else
\bibliography{bibliography}
\appendix
\fi
\section{\texorpdfstring{$\ceil{\tfrac{4k}{3}-\tfrac{5}{3}} d(u,v)$ approximation in $\Ot(mn^{1/k} + n^{1+2/k})$ time in unweighted graphs}{(4k/3-5/3) d(u,v) approximation in O(mn^(1/k) + n^(1+2/k)) time in unweighted graphs}}\label{S-k-stretch-small-k}

In this section, we present a tighter analysis of the algorithm from \Cref{S-4k/3-unweighted} (\Cref{A-npsp-unweighted-4k/3}), which yields an improved approximation bound.
We show:

\Reminder{T-k-stretch-k=4-5}

In the following lemma, we prove that the algorithm from \Cref{S-4k/3-unweighted} has a tighter guarantee on almost all the cases. To solve the remaining case (\Cref{sc-d=3i_v=0}) we set $\hat{d}(u,v)$ to $\min(\hat{d}(u,v), \min(d(u,w)+d(w,v) \mid w\in (\cup_{u_1\in B_0(u)\cup C(u,A_1)} B(u_1) \cup C(u_1,A_1)) \cap (B_0(v) \cup C_0(v))))$ in $\Ot(n^{1+2/k})$ time.
\begin{lemma}
    $\hat{d}(u,v) \le \ceil{\tfrac{4k}{3}-\tfrac{5}{3}} d(u,v)$.
\end{lemma}
\begin{proof}
    Let 
    \begin{equation}\label{eqq:t=d-u-v-/2+1}
        t=d(u,v)/2+0.5\ODD,
    \end{equation}
    and let $\tau\in P(u,v)$ be such that $d(u,\tau)=t$ and $d(v,\tau)=t-1\ODD$.
    Let $i_u$ be the first index such that $p_{i_u}(u)\in B(\tau)$ or $p_{i_u}(\tau)\in B(u)$, similarly, let $i_v$ be the first index such that $p_{i_v}(v)\in B(\tau)$ or $p_{i_v}(\tau)\in B(v)$.
    
    From the minimality of $i_u$, we have that for every $0\le i<i_u$ it holds that $p_{i}(u)\notin B(\tau)$ and $p_{i}(\tau)\notin B(u)$. (The same holds also for $v$.) Therefore, we can apply \Cref{L-Bound-p-i-d} to get: 
    \begin{equation}\label{eq2:hiu-iut}
        h_{i_u}(u),h_{i_u}(\tau) \le i_u\cdot d(u,\tau)=i_u\cdot t
    \end{equation}
    Similarly, we get that:
    \begin{equation}\label{eq2:hiv-ivt-1}
        h_{i_v}(v),h_{i_v}(\tau) \le i_v\cdot d(v,\tau)=i_v\cdot (t-1\ODD)
    \end{equation}
    
    Next we prove the first bound on $\hat{d}(u,v)$ and show that $\hat{d}(u,v) \le d(u,v) + 2i_u t + 2i_v(t-1\ODD)$. 
    \begin{cclaim}\label{eq:bound-du-v-p-iu-piv}
        $\hat{d}(u,v) \le d(u,v) + 2i_u t + 2i_v(t-1\ODD)$
    \end{cclaim}
    \begin{proof}
        Assume, wlog, that $p_{i_u}(u)\in B(\tau)$ and $p_{i_v}(\tau)\in B(v)$ (the other cases follow using similar methods).
        Since $p_{i_u}(u), p_{i_v}(\tau)\in B(\tau)$ we have that: 
        \begin{equation}\label{eq2:bound-H}
            H(p_{i_u}(u), p_{i_v}(\tau)) \le d(p_{i_u}(u),\tau)+h_{i_v}(\tau)
        \end{equation}
        Since $p_{i_u}(u)\in B(u)$ and $p_{i_v}(\tau)\in B(v)$ in the final loop of the algorithm there is an iteration in which $w=p_{i_u}(u)$ and $z=p_{i_v}(\tau)$. Thus, after this iteration, we have that:
        \begin{align*} 
        \begin{split}
            \hat{d}(u,v) &\le d(u,p_{i_u}(u))+H(p_{i_u}(u), p_{i_v}(\tau))+d(p_{i_v}(\tau),v) \\&\stackrel{\ref{eq2:bound-H}}\le d(u,v) + (d(p_{i_u}(u),\tau)+h_{i_v}(\tau)) +d(p_{i_v}(\tau),v) \stactri\le d(u,v) + 2h_{i_u}(u) + 2h_{i_v}(\tau)
            \\&\stackrel{\ref{eq2:hiu-iut},\ref{eq2:hiv-ivt-1}}\le d(u,v) + 2i_u t + 2i_v(t-1\ODD)
        \end{split}
        \end{align*}
    \end{proof}

    Notice that since $\hat{d}(u,v) \le d(u,v) + 2i_u t + 2i_v(t-1\ODD)$, the case that $i_u \ge i_v$ is harder than the case that $i_v < i_u$. Therefore, for the rest of the proof, we assume that $i_u\ge i_v$. Let $c\ge0 $ be the value such that $i_v=i_u-c$. 
    From \Cref{eq:bound-du-v-p-iu-piv} we get that:
    \begin{align}\label{e-1.1}
        \begin{split}
        \hat{d}(u,v) &\stackrel{\ref{eq:bound-du-v-p-iu-piv}}\le d(u,v) + 2i_u t + 2i_v(t-1\ODD) \stackrel{i_v=i_u-c}=d(u,v) + 2i_u t + 2i_u (t-1\ODD) -c(t-1\ODD)
        \\&\stackrel{\ref{eqq:t=d-u-v-/2+1}}=d(u,v)(1+2i_u-c/2)+(c/2)\ODD
        \end{split}
    \end{align}
    
    Next, we prove the following claim that either bounds $\min(h_{i_u+1}(u), h_{i_u+1}(v))$ or bounds $\hat{d}(u,v)$.
    \begin{cclaim}\label{C-bound-h-i+1}
        One of the following conditions holds:
        \begin{itemize}
            \item $\min(h_{i_u+1}(u), h_{i_u+1}(v)) \le (i_u+1)d(u,v)/2 + (c/2-i_v+0.5)\cdot 1\ODD$
            \item $\hat{d}(u,v) \le kd(u,v)+(c-2i_v)\ODD$
        \end{itemize}
    \end{cclaim}
    \begin{proof}
    From the definition of $i_u$ it follows that for every $i_v< i< i_u$ it holds that $p_i(\tau)\notin B(u)$ and $p_i(u)\notin B(\tau)$. Therefore we can apply \Cref{L-Bound-p-i-d} and get that:
    \begin{align}\label{eq_g:h_i(u)h_i(v)}
    \begin{split}
        h_{i_u}(u),h_{i_u}(\tau) &\le (i_u-i_v)d(u,\tau)+h_{i_v}(\tau) \stackrel{\ref{eq2:hiv-ivt-1}}\le (i_u-i_v)t+i_v\cdot (t-1\ODD)
        \\&\stackrel{i_v=i_u-c}=ct+(i_u-c)(t-1\ODD)
        \\&\stackrel{\ref{eqq:t=d-u-v-/2+1}}
        =i_u\cdot d(u,v)/2+(c/2-i_v)\cdot 1\ODD
    \end{split}
    \end{align}

    If $p_{i_u}(\tau)\in B(u)\cap B(v)$, then in the final loop in the algorithm there is an iteration in which $w=z=p_{i_u}(\tau)$, and after this iteration we have that:
    \begin{equation*}
        \hat{d}(u,v) \le d(u,v) + 2h_{i_u}(\tau) \stackrel{\ref{eq_g:h_i(u)h_i(v)}}\le d(u,v)+ 2(i_u\cdot d(u,v)/2+(c/2-i_v)\cdot 1\ODD) \stackrel{i_u\le k-1}\le kd(u,v)+(c-2i_v)\ODD,
    \end{equation*}
    as required. 
    Otherwise, we have that $p_{i_u}(\tau)\notin B(u)$ or $p_{i_u}(\tau)\notin B(v)$, and therefore:
    \begin{align*}\label{eqq:min-h_i_v+1}
        \min(h_{i_u+1}(u), h_{i_u+1}(v)) &\le h_{i_u}(\tau)+t \stackrel{\ref{eq_g:h_i(u)h_i(v)}}\le i_u\cdot d(u,v)/2+(c/2-i_v)\cdot 1\ODD + t
        \\&\stackrel{\ref{eqq:t=d-u-v-/2+1}} = (i_u+1)d(u,v)/2 + (c/2-i_v+0.5)\cdot 1\ODD,
    \end{align*}
    as required.
    \end{proof}

    Armed with bounds \Cref{eq:bound-du-v-p-iu-piv} and \Cref{C-bound-h-i+1} we are ready to prove our lemma. We divide the proof of the lemma into two cases: the case where $c/2-i_v \le -1$ and the case where $c/2-i_v\ge 0$.
    We start by proving the case where $c/2-i_v \le -1$ in the following claim.
    \begin{cclaim}
        If $c/2-i_v\le -1$ then $\hat{d}(u,v) \le \ceil{\tfrac{4k}{3}-\tfrac{5}{3}}d(u,v)$
    \end{cclaim}
    \begin{proof}
    Notice, that if $i_u \le \ceil{\tfrac{2k}{3} - \tfrac{4}{3}}$, then we get that 
    \begin{align*}
        \hat{d}(u,v) \stackrel{\ref{e-1.1}}\le (1+2i_u)d(u,v) -c(t-1\ODD) \stackrel{c\ge 0}\le (1+2i_u)d(u,v) 
        \\\stackrel{i_u \le \ceil{\tfrac{2k}{3} - \tfrac{4}{3}}}\le (1+2\ceil{\tfrac{2k}{3} - \tfrac{4}{3}})d(u,v) \le \ceil{\tfrac{4k}{3}-\tfrac{5}{3}}d(u,v),
    \end{align*}
    as required.
    Therefore, for the rest of the proof of this case, we assume that $i_u > \ceil{\tfrac{2k}{3} - \tfrac{4}{3}}$, since $i_u$ and $\ceil{\tfrac{2k}{3} - \tfrac{4}{3}}$ are integers, we get that:
    \begin{equation}\label{eqq:i_u-big-1}
        i_u \ge \ceil{\tfrac{2k}{3} - \tfrac{4}{3}} + 1=\ceil{\tfrac{2k}{3} - \tfrac{1}{3}}
    \end{equation}
    
    From \Cref{CC-hat-d-u-v-le-thzq} we have that $\hat{d}(u,v) \le \THZQuery(u,v)$, and therefore from \Cref{L-THZ05-Q-Correctness-private} with parameter $(i_u+1)$ we get that
    \begin{align*}
        \hat{d}(u,v) &\le d(u,v)+2(k-1-i_u-1)d(u,v)+2\min(h_{i_u+1}(u), h_{i_u+1}(v))
        \\&\stackrel{C\ref{C-bound-h-i+1}}\le (1+2k-2-2i_u-2)d(u,v)+2((i_u+1)d(u,v)/2 + (c/2-i_v+0.5)\cdot 1\ODD)
        \\&=(2k-2-i_u)d(u,v) + (c-2i_v+1)\cdot 1\ODD \stackrel{c/2-i_v\le -1}\le (2k-2-i_u)d(u,v)
        \\&\stackrel{\ref{eqq:i_u-big-1}}\le (2k-2-\ceil{\tfrac{2k}{3} - \tfrac{1}{3}})d(u,v) \le \ceil{\tfrac{4k}{3}-\tfrac{5}{3}}d(u,v),
    \end{align*}
    as required. 
    Notice that from \Cref{C-bound-h-i+1} we might have that $\hat{d}(u,v) \le kd(u,v)+(c-2i_v)\ODD$, but since $c/2-i_v<0$ we get that in this case $\hat{d}(u,v) \le kd(u,v) \le \ceil{\tfrac{4k}{3}-\tfrac{5}{3}}d(u,v)$, as required.
    \end{proof}
    Next, we consider the case that $c/2-i_v\ge 0$.
    \begin{cclaim}
        If $c/2-i_v\ge 0$ then $\hat{d}(u,v) \le \ceil{\tfrac{4k}{3}-\tfrac{5}{3}}d(u,v)$
    \end{cclaim}
    \begin{proof}
         Notice, that if $i_u \le 2k/3-4/3+c/4 + (c/4)\cdot 1\ODD/d(u,v)$, then:
    \begin{align*}
    \begin{split}
        \hat{d}(u,v) &\stackrel{\ref{e-1.1}}\le (1+2i_u)d(u,v)-c(t-1\ODD)=(1+2i_u-c/2)d(u,v)-(c/2)\cdot 1\ODD 
        \\&\le (1+2(\tfrac{2k}{3}-\frac{4}{3}+c/4 + (c/4)\cdot 1\ODD/d(u,v))-c/2)d(u,v)-(c/2)\cdot 1\ODD
        \\&=(\tfrac{4k}{3}-\tfrac{5}{3})d(u,v),
    \end{split}
    \end{align*}
    as required. Therefore, for the rest of the proof, we can assume that:
    \begin{equation}\label{eqqqqqqq:i_u-big}
        i_u > \tfrac{2k}{3}-\tfrac{4}{3}+c/4 + c/4\cdot 1\ODD/d(u,v)
    \end{equation}
    
    From \Cref{CC-hat-d-u-v-le-thzq} we have that $\hat{d}(u,v) \le \THZQuery(u,v)$, and therefore from \Cref{L-THZ05-Q-Correctness-private} with parameter $(i_u+1)$ we get that:
    \begin{align}\label{eqaqa-bound-hatduv}
    \begin{split}
    \hat{d}(u,v) &\le d(u,v)+2(k-1-i_u-1)d(u,v)+2\min(h_{i_u+1}(u), h_{i_u+1}(v))
    \\&\stackrel{C\ref{C-bound-h-i+1}}\le (1+2k-2-2i_u-2)d(u,v) + 2((i_u+1)d(u,v)/2+(c/2-i_v+0.5)\cdot 1\ODD)
    \\&\stackrel{i_v=i_u-c}= (2k-2-i_u)d(u,v) + (3c-2i_u+1)\cdot 1\ODD
    \end{split}
    \end{align}

    Notice that in \Cref{C-bound-h-i+1} we might have that $\hat{d}(u,v) \le kd(u,v)+(c-2i_v)\ODD$, to have that $kd(u,v)+(c-2i_v)\ODD \le \ceil{\tfrac{4k}{3}-\tfrac{5}{3}}d(u,v)$, we need to have that either $d(u,v) \ge 4$ or $c<i_u$. Therefore, we divide the proof into three cases. The case that $d(u,v) \ge 4$, the case that $d(u,v)\le 3$ and $c<i_u$, and the case that $d(u,v)\le 3$ and $c=i_u$.
    \begin{subclaim}
        If $d(u,v) \ge 4$ then $\hat{d}(u,v)\le \ceil{\tfrac{4k}{3}-\tfrac{5}{3}}d(u,v)$
    \end{subclaim}
    \begin{proof}
    Since $d(u,v) \ge 4$, we have that $1\ODD \le d(u,v)/5$. Therefore:
    \begin{align*}
        \hat{d}(u,v) &\le (2k-2-i_u)d(u,v) + (3c-2i_u+1)\cdot 1\ODD
        \\&\stackrel{1\ODD \le d(u,v)/5}\le (2k-2-i_u+11c/20+1/5-2i_u/5)d(u,v) + (c/4)\cdot 1\ODD
        \\&=(2k-9/5-7i_u/5+11c/20)d(u,v) + (c/4)\cdot 1\ODD
        \\&\stackrel{c\le i_u}\le (2k-9/5-i_u+3c/20)d(u,v)
        \\&\stackrel{\ref{eqqqqqqq:i_u-big}}\le (2k-9/5-(\tfrac{2k}{3}-\tfrac{4}{3}+c/4 + c/4\cdot 1\ODD/d(u,v))+3c/20)d(u,v)  + (c/4)\cdot 1\ODD \\&\le (\tfrac{4k}{3}-\tfrac{7}{15}-c/10)
        \le \tfrac{4k}{3}-\tfrac{7}{15}
    \end{align*}
    Since $\hat{d}(u,v)$ is an integer, we get that $\hat{d}(u,v) \le \floor{\tfrac{4k}{3}-\tfrac{7}{15}}$, since $k$ is an integer we get that 
    $\floor{\tfrac{4k}{3}-\tfrac{7}{15}}\le \tfrac{4k}{3}-\tfrac{5}{3}$, and overall we get that $\hat{d}(u,v)\le \tfrac{4k}{3}-\tfrac{5}{3}$, as required.
    Notice that in \Cref{C-bound-h-i+1} we might have that:
    $$\hat{d}(u,v) \le kd(u,v)+(c-2i_v)\ODD \stackrel{c\le k-1}\le kd(u,v)+(k-1)\ODD \stackrel{1\ODD \le d(u,v)/5} \le 6/5kd(u,v),$$ as required.       
    \end{proof}
    
    Next, we consider the case where $d(u,v) \le 3$ and $c<i_u$.
    \begin{subclaim}
        If $d(u,v) \le 3$ and $c<i_u$ then $\hat{d}(u,v) \le \ceil{\tfrac{4k}{3}-\tfrac{5}{3}}d(u,v)$
    \end{subclaim}
    \begin{proof}
    We can solve the case that $d(u,v)=1$ in $O(n)$ time by checking whether $(u,v)\in E$. If $d(u,v)=2$ then we have that $1\ODD=0$, and the claim holds using the same methods. 
    If $d(u,v)=3$ we get that:
    \begin{align}\label{eqq3:temp3:3}
    \begin{split}
    \hat{d}(u,v)&\stackrel{\ref{eqaqa-bound-hatduv}}\le (2k-2-i_u)d(u,v) + (3c-2i_u+1)\cdot 1\ODD
    \\&\stackrel{d(u,v)=3}=6k-6-3i_u+3c-2i_u+1=6k-5i_u+3c-5
    \end{split}
    \end{align}
    In this case, if $i_u \le (4k-5)/6+c/3$, then:
    \begin{align*}
        \hat{d}(u,v)&\stackrel{\ref{e-1.1}}\le (1+2i_u-c/2)d(u,v)-(c/2)\cdot 1\ODD\stackrel{d(u,v)=3}=3+6i_u-2c \stackrel{i_u\le (4k-5)/6+c/3}\le (4k-5)
        \\&\stackrel{d(u,v)=3}=(\tfrac{4k}{3}-\tfrac{5}{3})\cdot d(u,v), \text{ as required.}
    \end{align*}
    Therefore, for the rest of the proof, we assume that $i_u > (4k-5)/6+c/3$. Thus:
    \begin{align*}
    \hat{d}(u,v) &\le 6k-5i_u+3c-5 < 6k-3((4k-5)/6+c/3)-2i_u+3c-5 \le 4k-2i_u+2c-2.5
    \\&\stackrel{c\le i_u-1}\le 4k-4.5
    \end{align*}
    Since $k$ and $\hat{d}(u,v)$ are integers, having $\hat{d}(u,v) < 4k-4.5$ guarantees that $\hat{d}(u,v) \le 4k-5$, as required since $(\tfrac{4k}{3}-\tfrac{5}{3})d(u,v)=4k-5$.
    \end{proof}
    
    Next, we complete the proof by considering the special case where $d(u,v)=3$ and $c=i_u$.
    \begin{subclaim}\label{sc-d=3i_v=0}
        If $d(u,v)=3$ and $c=i_u$ then $\hat{d}(u,v) =d(u,v) \le \ceil{\tfrac{4k}{3}-\tfrac{5}{3}}d(u,v)$
    \end{subclaim}
    \begin{proof}
        From the definition of $c$, in this case, we have that $i_v=0$.
        Let $P(u,v)=(u,u_1,v_1,v)$, since $i_v=0$, we have that $v_1\in B_0(v)\cup C(v,A_1)$, and from symmetry we can also get that $u_1\in B_0(u)\cup C(u,A_1)$, and using similar arguments we get that $v_1\in B(u_1)\cup C(u_1,A_1)$.
        
        Therefore, we get that $\min(d(u,w)+d(w,v) \mid w\in (\cup_{u_1\in B_0(u)\cup C(u,A_1)} B(u_1) \cup C(u_1,A_1)) \cap (B_0(v) \cup C_0(v)))=d(u,v)$, and therefore $\hat{d}(u,v)=d(u,v)$, as required.
    \end{proof}
    \end{proof}
\end{proof}

\end{document}